\documentclass[12pt]{article}
\usepackage{amsmath,bbm,amssymb}
\usepackage{amsthm}
\usepackage{comment}
\usepackage{graphicx, graphics}
\usepackage{textcomp}
\usepackage{amsfonts}
\usepackage{psfrag,pst-node,subfigure,rotating, amsmath, bbm, amsthm, amssymb, amsthm, setspace, picture, epsfig, amsfonts, upgreek}
\usepackage{multirow}
\usepackage{arydshln}
\usepackage{algorithm}
\usepackage{algorithmic}
\usepackage{mathtools}
\usepackage{arydshln}

\newtheorem{assumption}{Assumption}
\newtheorem{theorem}{Theorem}
\newtheorem{lemma}{Lemma}

\newtheorem{definition}{Definition}

\newcommand{\vg}[1]{\mbox{\boldmath$ #1 $\unboldmath}}

\def\v{\mathbf}\def\m{\mathbf}\def\vg{\boldsymbol}\def\s{\mathcal}

\begin{document}
	\begin{center}
		{\Large\bf Linear screening for high-dimensional computer experiments}
		\\[2mm] Chunya \textsc{Li}$^{1,2}$, Daijun \textsc{Chen}$^3$, Shifeng \textsc{Xiong}$^{2}$\footnote{Corresponding author: Shifeng Xiong. Email address: xiong@amss.ac.cn.}
		\\ 1.School of Mathematical Sciences, University of Chinese Academy of Sciences\\ 2.NCMIS, Academy of Mathematics and Systems Science \\Chinese Academy of Sciences\\3.Nuance Communications Inc
	\end{center}
	
	\vspace{1cm} \noindent{\bf Abstract}\quad
	In this paper we propose a linear variable screening method for computer experiments when the number of input variables is larger than the number of runs. This method uses a linear model to model the nonlinear data, and screens the important variables by existing screening methods for linear models. When the underlying simulator is nearly sparse, we prove that the linear screening method is asymptotically valid under mild conditions. To improve the screening accuracy, we also provide a two-stage procedure that uses different basis functions in the linear model. The proposed methods are very simple and easy to implement. Numerical results indicate that our methods outperform existing model-free screening methods.
	
	
	\vspace{1cm} \noindent{{\bf KEY WORDS:} Best linear approximation, best subset regression, nonlinear model, sensitivity analysis, sure independence screening.}

	
	
	\section{Introduction}\label{sec:intro}
	
	Nowadays computer experiments are commonly used to study computer simulations in engineering and scientific investigations (Santner, Williams, and Notz 2018). Computer simulations usually have complex nonlinear input-output relationships with long running times. Furthermore, they often involve larger numbers of input variables (Fang, Li, and Sudjianto 2006).
	For example, building performance simulation is used to predict performance aspects of a building, and its inputs include various types of parameters such as climate parameters, geometry parameters, envelope parameters, and so on. For large buildings, the number of these inputs can be much larger than one hundred (Clarke 2001). Examples of computer simulations with large numbers of input variables can also be found in climate simulations (Roulstone and Norbury 2013) and manufacturing simulations (Jahangirian et al. 2010).
	
	Many authors discussed the screening/selection problem or related sensitivity analysis problem for computer simulations with many inputs. If only a small proportion of the inputs are active or influential, variable screening or sensitivity analysis methods can detect active inputs that have major impact on the output, and thus we can better understand the input-output relationship. Morris (1991) proposed a design-based one-factor-at-a-time factor screening method. Schonlau and Welch (2006) presented a screening method via analysis of variance and visualization. Linkletter et al. (2006) and Reich, Strolie, and Bondell (2009) provided Bayesian selection methods. Moon, Dean, and Santner (2012) proposed a two-stage sensitivity-based group screening method. Sung et al. (2017) provided a multi-resolution functional ANOVA approach for many-input computer experiments. However, these methods are not applicable to the cases where the number of variables is larger than the number of runs. Such cases are common in practice since we usually have limited runs to analyze a high-dimensional computer simulation due to the long running time. In addition, the ``large $p$ small $n$" problem often appears in the first stage of analyzing high-dimensional simulations. Based on the screening result in the first stage, more efficient design and analysis strategies can be made in the follow-up study.
	
	This paper focuses on the variable screening problem for computer experiments when the number of inputs, $p$, is larger than the number of runs, $n$. In recent years, plenty of methodologies were proposed to screen important variables for $p>n$ problems in statistics. Fan and Lv (2008) proposed the sure independence screening (SIS) method for linear regression models. This method was extended to generalized linear models (Fan and Song 2010), nonparametric additive models (Fan, Feng, and Song 2011), and varying coefficient models (Fan, Ma, and Dai 2014). Many model-free screening methods were also provided in the literature; see, e.g., Zhu et al. (2011), Li, Zhong, and Zhu (2012), Huang and Zhu (2016), and Lu and Lin (2017). These model-free methods can be used for aforementioned high-dimensional computer experiments, and their performance is in need of evaluation.
	
	It should be noted that most screening methods for $p>n$ cases in the literature are marginal methods that only use the separate relationship between each variable and the response. In this paper we consider the screening problem from another angle. Compared to the number of variables, available runs are very limited. It seems that models as simple as possible should be first considered. Therefore, we adopt the linear regression model to the data from high-dimensional computer experiments, and use the $\ell_0$-screening principle for the linear model (Xiong 2014; Xu and Chen 2014) to screen the active input variables of the nonlinear simulator. It can be seen that the idea of this linear screening method is similar to that of the regression method in global sensitivity analysis for computer experiments (Santner, Williams, and Notz 2018), which uses regression coefficients under the linear regression model as sensitivity indices for the input variables.

The linear screening method is very simple and easy to implement. One of the main contributions of this paper is to prove its asymptotic validity.
To handle the bias cased by the model simplicity, we investigate the best linear approximation (BLA) of a nonlinear computer simulator. When the simulator is nearly sparse, we show that the active variables are still active in its BLA under mild conditions. Based on this, we prove the asymptotic validity of the linear $\ell_0$-screening principle for computer experiments with $p>n$. Consequently, sophisticated screening algorithms other than the marginal methods for linear regression models are proposed in our linear screening procedure. A large number of numerical results indicate that the proposed methods perform better than the marginal screening methods in the literature. In addition, the screening accuracy of the proposed methods can be improved through using different basis functions in the underlying linear model.

	The rest of the paper is organized as follows. In section 2, we give the definition of BLA and discuss its properties. Section 3 provides theoretical results on the asymptotic validity of the linear screening methods for nonlinear computer models. In section 4, we discuss the linear screening methods with different basis functions. Section 5 gives the numerical results. Section 6 ends this paper with some discussion. Additional definitions and all proofs are given in the Supplementary Materials.

	\section{Best linear approximation of a nonlinear function}
	
	Suppose that the input-output relationship of a deterministic computer simulation is\begin{equation} y=f(\v{x}),\label{f}\end{equation}where the input variables $\v{x}=\left(x_1,\ldots,x_p\right)'\in[0,1]^p$, $f$ is continuous, i.e., $f\in C[0,1]^p$, and $'$ denotes the transpose. For $n$ design sites $\v{x}_1,\ldots,\v{x}_n$, the corresponding outputs are $y_1,\ldots,y_n$, where $\v{x}_i=\left(x_{i1},\ldots,x_{ip}\right)'\in[0,1)^p$. Let $\m{X}=\left(\v{x}_1,\cdots,\v{x}_n\right)'$ and $\v{y}=\left(y_1,\ldots,y_n\right)'$. When we discuss asymptotics, $f$ in \eqref{f} depends on $n$, and is also written as $f_n$.
	
	When $p$ is larger than $n$, popular modeling and variable selection methods for computer experiments such as Kriging (Matheron 1963) are difficult to apply. Compared with the dimensionality, our data are very limited. It seems that we should use a very simple model for the data. Here the linear regression model \begin{equation}\label{lm}y=\phi_0+\vg{\phi}'\v{x}+\epsilon\end{equation}is under consideration, where $\phi_0 \in \mathbb{R},\ \vg{\phi} \in \mathbb{R}^p$ are unknown coefficients and $\epsilon$ is the measurement error. In fact, the linear part of the above linear model corresponds to the best linear approximation (BLA) of $f$, which is defined as\begin{align*}
	\beta_0+\vg{\beta}'\v{x}
	=\mathop{\arg\min}_{g\in \{\phi_0+{\vg{\phi}}'\v{x}:\ \phi_0 \in \mathbb{R},\ \vg{\phi} \in \mathbb{R}^p\}}\int_{[0,1]^p}\left[f(\v{x})-g(\v{x})\right]^2d\v{x}.
	\end{align*}
	Let $\vg{\beta}=\left(\beta_1,\cdots,\beta_p\right)'$ and
	\begin{align*}H(\phi_0,\phi_1,\cdots,\phi_p)&=\int_{[0,1]^p}\left[f\left(\v{x}\right)-\left(\phi_0+\phi_1x_1+\ldots+\phi_px_{p}\right)\right]^2d\v{x}.
	\end{align*}
	Taking partial derivatives of $H$ with respect to $\phi_0, \ldots, \phi_p$ and letting them be equal to zero, we have
	\begin{eqnarray}
	&&\beta_0=\int_{[0,1]^p}f_n\left(\v{x}\right)d\v{x}-\frac{1}{2}\sum\limits_{j=1}^p\beta_j,\label{beta0}\\
	&&\beta_j=12\left(\int_{[0,1]^p}x_jf_n\left(\v{x}\right)dx-\frac{1}{2}\int_{[0,1]^p}f_n\left(\v{x}\right)dx\right),\ j=1,\cdots,p.\label{betaj}
	\end{eqnarray}

	To discuss theoretical properties of our methods, we make the following basic assumption.
	\begin{assumption}\label{assumption 2.1}
		There exist $\widetilde{f}_n \in C[0,1]^{p_0}$ and $\eta_n>0$ such that
		\begin{align}\label{2}
		\sup_{\left(x_1,\ldots,x_p\right)'\in[0,1]^p}\left|f_n(x_1,\ldots,x_p)-\widetilde{f}_n\left(x_1,\cdots, x_{p_0}\right)\right| < \eta_n,
		\end{align}
		where $p_0<p$. Furthermore, for each $j=1,\cdots,p_0$,
		\begin{align}\label{3}
		\left|\int_{[0,1]^p} x_jf_n\left(\v{x}\right)d\v{x}-\frac{1}{2}\int_{[0,1]^p}f_n\left(\v{x}\right)d\v{x}\right| >\tau
		\end{align}  where $\tau$ is a positive constant.
	\end{assumption}
	Unlike the (complete) sparsity assumption in the literature of high-dimensional screening, we allow the computer model $f_n$ to be different from a model with $p_0$ variables in the first part \eqref{2} of Assumption \ref{assumption 2.1}. This indicates that $f_n$ also depends on the less important variables $x_{p_0+1}\cdots,x_{p}$, and matches the practical cases better than the sparsity assumption. By \eqref{betaj}, the second part \eqref{3} requires that each of these $p_0$ variables should be active in the BLA of $f_n$. Specially, we have the following result.
	\begin{theorem}\label{THM1} Under Assumption \ref{assumption 2.1}, $\left|\beta_j\right|<12\eta_n$ for $j=p_0+1,\ldots,p$, and $\left|\beta_j\right|>12\tau$ for $j=1,\ldots,p_0$.
	\end{theorem}
	
	For an integer $d$, let $Z_d=\left\{1,\cdots,d\right\}$. If $\eta_n\to0$ as $n\to\infty$, then Theorem \ref{THM1} indicates that, under Assumption \ref{assumption 2.1}, $\s{A}_0=Z_{p_0}$ can be viewed as the true submodel of both the original model $f_n$ and its BLA. Note that the vector of coefficients $\vg{\beta}$ can be sparsely estimated by regularized least squares under the linear model \eqref{lm}. Theorem \ref{THM1} basically guarantees the validity of our linear screening methods to select $\s{A}_0$. Further discussion on \eqref{3} can be seen in Section \ref{sec:db}.

	\section{Asymptotic validity of linear screening}\label{sec:ls}

	Given a pre-specified integer $M$ with $p_0\leqslant M\ll n$, our purpose is to find a $M$-subset of $Z_p$ that includes the true submodel $\s{A}_0$ of $f_n$ in \eqref{f}. From the discussion in the previous section, $\s{A}_0$ is also the true submodel of the BLA of $f_n$ under certain conditions. This inspires us to use screening methods for linear models to the nonlinear model $f_n$.
	
	Some definitions and notation are needed here. For a vector $\v{x}$, $\left\|\v{x}\right\|$ and $\left\|\v{x}\right\|_0$ denote its Eucildean norm and $\ell_0$ norm, respectively. For a matrix $\m{A}$, $\|\m{A}\|$ denotes its spectral norm. For a set $\s{S}$, $|\s{S}|$ denotes its cardinality. For $\s{A}\subset Z_p$, let $\v{x}_\s{A}$ and $\m{X}_\s{A}$ denote the subvector of $\v{x}$ and the submatrix of $\m{X}$ corresponding to $\s{A}$, and define
	\begin{align}\label{betaA}
	\beta_0(\s{A})+\vg{\beta}\left(\s{A}\right)'\v{x}_{\s{A}}=\mathop{\arg\min}_{g\in \{\phi_0+{\vg{\phi}}'\v{x}_{\s{A}}:\ \phi_0 \in \mathbb{R},\ \vg{\phi} \in \mathbb{R}^{|A|}\}}\int_{[0,1]^p}\left[f(\v{x})-g(\v{x}_\s{A})\right]^2d\v{x}.
	\end{align}
	Note that $\vg{\beta}(\s{A})=(\beta_1(\s{A}),\cdots,\beta_{|\s{A}|}(\s{A}))'$ is a $|\s{A}|$-dimensional vector. Let
	\begin{align}\label{betaZp}
	\vg{\beta}_{Z_p}(\s{A})=({\beta}_{Z_p}(\s{A})_1,\cdots,{\beta}_{Z_p}(\s{A})_p)'
	\end{align} be the $p$-dimensional vector obtained by expanding $\vg{\beta}(\s{A})$ with
	$\vg{\beta}_{Z_p}(\s{A})_{\s{A}}=\vg{\beta}(\s{A})$ and $\vg{\beta}_{Z_p}(\s{A})_j=0$ for $j\notin\s{A}$.
	Based on the data $(\m{X},\v{y})$ generated from \eqref{f}, for $|\s{A}|<n$, $(\beta_0(\s{A}),\vg{\beta}(\s{A})')'$ in \eqref{betaA} can be estimated by the least squares method under the linear model \eqref{lm},
	\begin{align}\label{1}
	\begin{pmatrix}
	\widehat{\beta}_0(\s{A})\\
	\widehat{\vg{\beta}}(\s{A})
	\end{pmatrix}=\begin{pmatrix}
	n&&\vg{1}_n'\m{X}_\s{A}\\
	\m{X}_\s{A}'\vg{1}_n&&\m{X}_\s{A}'\m{X}_\s{A}
	\end{pmatrix}^{-1}\begin{pmatrix}
	\vg{1}_n' \v{y}\\\m{X}'_\s{A} \v{y}
	\end{pmatrix},
	\end{align}and $\widehat{\vg{\beta}}_{Z_p}(\s{A})$ can be defined similarly.
	
	Let $V_{HK}(f)$ denote the Hardy-Krause variation (see Definition A.1 in the Supplementary Materials) of $f$. If $V_{HK}(f)<\infty$, then we say that $f$ has bounded variation in Hardy-Krause sense (BVHK), also write $f\in \mathrm{BVHK}$. For a design $\m{A}=(\v{a}_1,\cdots,\v{a}_n)'$ with $\v{a}_1,\ldots,\v{a}_n \in [0,1]^d$, let $\delta_{d,n}(\m{A})$ denote its $L_\infty$ discrepancy (see Definition A.2 in the Supplementary Materials).
	
	\begin{assumption}\label{assumption 2.2} For each $n$, $\widetilde{f}_n\in \mathrm{BVHK}$.
	\end{assumption}
	If $\widetilde{f}_n$ is sufficiently smooth on $[0,1]^{p_0}$, then Assumption \ref{assumption 2.2} holds (see Lemma 3 in the Supplementary Materials).
	Define\begin{align*}
	&\s{U}_0=\{\s{A}\subset Z_p:|\s{A}|=M,\s{A}_0\subset \s{A}\},\ \s{U}_1=\{\s{A}\subset Z_p:|\s{A}|=M,\s{A}_0 \setminus  \s{A}\neq \emptyset \},\\
	&\s{U}_2=\{\s{A} \subset Z_p: |\s{A}|=2\},\ \s{U}_3=\{\s{A} \subset Z_p: \s{A}=\s{A}_0 \cup I_1,|I_1|=1, I_1 \in Z_p\setminus \s{A}_0\},\\
	&\s{U}_4=\{\s{A} \subset Z_p: \s{A}=\s{A}_0 \cup I_2,|I_2|=M-1, I_2 \in Z_p\setminus \s{A}_0\},\\
	&\Delta_{p_0,n}=\delta_{p_0,n}(\m{X}_{\s{A}_0}),\ \Delta_{2,n}=\max_{\s{A} \in \s{U}_2}\delta_{2,n}(\m{X}_\s{A}),\\
	&\Delta_{p_0+1,n}=\max_{\s{A} \in \s{U}_3}\delta_{p_0+1,n}(\m{X}_\s{A}),\ \Delta_{p_0+M-1,n}=\max_{\s{A} \in \s{U}_4}\delta_{p_0+M-1,n}(\m{X}_\s{A}).
	\end{align*}
	
	\begin{assumption} \label{assumption 2.4}
		There exists a constant $\alpha \in (0,1)$ such that $3(M+1)^2(9M+1)\Delta_{2,n}<\alpha$ for sufficiently large $n$.
	\end{assumption}	
	
	Write \begin{align*}
	&C_{1n}=\max\limits_{j=1,\cdots,p}\left(V_{HK}(\widetilde{f}_n),V_{HK}(x_j\widetilde{f}_n)\right),\ C_{2n}=\max\limits_{\v{x} \in [0,1]^p} \left|f_n(\v{x})\right|,
	\\	&C_{3n}=\max\limits_{\v{x} \in [0,1]^p}\left|\widetilde{f}_n(\v{x}_{\s{A}_0})-\v{x}'\vg{\beta}_{Z_p}({\s{A}_0})-\beta_0(\s{A}_0)\right|,\end{align*}
	 \begin{align*}&C_{4n}=\max\limits_{\s{A} \in \s{U}_1 }\max\limits_{\v{x} \in [0,1]^p}\left|\widetilde{f}_n(\v{x}_{\s{A}_0})-\v{x}'\vg{\beta}_{Z_p}({\s{A}})-\beta_0(\s{A})\right|,\\
	&V_{1n}=V_{HK}\left(\widetilde{f}_n(\v{x}_{\s{A}_0})-\beta_0(\s{A}_0)-\v{x}'\vg{\beta}_{Z_p}(\s{A}_0)\right),\ V_{2n}=\max_{\s{A} \in \s{U}_1}V_{HK}\left(\widetilde{f}_n(\v{x}_{\s{A}_0})-\beta_0(\s{A})-\v{x}'\vg{\beta}_{Z_p}(\s{A})\right),\\
	&\zeta_{1n}
	=C_{1n}(p_0+1)^2(9p_0+1)\Delta_{p_0,n}+2(p_0+1)^2(9p_0+1)\eta_n+\frac{3}{1-\alpha}C_{2n}(p_0+1)^4(9p_0+1)^2\Delta_{2,n},\\
	&\zeta_{2n}
	=C_{1n}(M+1)^2(9M+1)\Delta_{p_0+1,n}+2(M+1)^2(9M+1)\eta_n+\frac{3}{1-\alpha}C_{2n}(M+1)^4(9M+1)^2\Delta_{2,n},\\
	&\rho_{1n}= \zeta_{1n}^2+2(C_{3n}+\eta_n)	\zeta_{1n}+2(p_0+1)^{-1/2}\zeta_{1n}^2,\ \rho_{2n}=\zeta_{2n}^2+2(C_{4n}+\eta_n)	\zeta_{2n}+2(M+1)^{-1/2}\zeta_{2n}^2.
	\end{align*}
	
	\begin{assumption} \label{assumption 2.5}
		There exists a constant $D>0$ such that $12\tau^2-12\left(M-p_0+1\right)\eta_n^2-4\eta_n^2-4\eta_n\left(C_{4n}+C_{3n}\right)-\Delta_{p_0,n}V_{1n}-\Delta_{p_0+M-1,n}V_{2n}-\rho_{1n}-\rho_{2n}>D$  for sufficiently large $n$.
	\end{assumption}

	\begin{theorem}\label{THM4}
		Under Assumptions \ref{assumption 2.1}, \ref{assumption 2.2}, \ref{assumption 2.4}, and \ref{assumption 2.5},  for sufficiently large $n$,\begin{align*}\max\limits_{\s{A}\in\s{U}_0}\left\|\v{y}-\vg{1}_n\widehat{\beta}_0(\s{A})-\m{X}\widehat{\vg{\beta}}_{Z_p}({\s{A}})\right\|^2
		<\min\limits_{\s{A}_1\in\s{U}_1}\left\|\v{y}-\vg{1}_n\widehat{\beta}_0(\s{A}_1)-\m{X}\widehat{\vg{\beta}}_{Z_p}({\s{A}_1})\right\|^2.
		\end{align*}
	\end{theorem}

	Xiong (2014) and Xu and Chen (2014)	presented similar results to Theorem \ref{THM4} for high-dimensional linear regression models: a subset that includes the true submodel always yields a smaller residual sum of squares than those that do not. Therefore, we can screening important variables in linear models through solving the $\ell_0$-constrained least squares problem
	\begin{equation}\label{bsr}\min_{\beta_0\in\mathbb{R},\ \vg{\beta}\in\mathbb{R}^p}\left\|\v{y}-\vg{1}_n\beta_0-\m{X}\vg{\beta}\right\|^2\quad{\text{subject to}}\ \left\|\vg{\beta}\right\|_0\leqslant M.\end{equation}
	A sub-optimal solution to \eqref{bsr} can still include the true submodel $\s{A}_0$ (Xiong 2014). The famous $\ell_1$-regularized method (lasso) (Tibshirani 1996) is a convex approximation to \eqref{bsr}. A number of papers provided efficient algorithms for solving \eqref{bsr} and showed that this $\ell_0$ method can be preferable over the $\ell_1$-regularized and other methods in variable selection/screening for linear models from theoretical and/or empirical aspects (Shen et al. 2013; Xiong 2014; Xu and Chen 2014; Bertsimas, King, and Mazumder 2016). Theorem \ref{THM4} indicates under certain conditions that the $\ell_0$ screening method for linear models is still effective for the nonlinear computer model \eqref{f}: when the residual sum of squares becomes small to some level, the corresponding subset includes the true submodel asymptotically. Hence, we propose to screen the true submodel $\s{A}_0$ of \eqref{f} based on algorithms for solving \eqref{bsr}.

	Assumptions \ref{assumption 2.4} and \ref{assumption 2.5} are not easy to verify in practice. The following theorem provides sufficient conditions for them.
	
	\begin{theorem} \label{thm:as}
		Suppose that $\Delta_{2,n}=\Delta_{p_0+1,n}=\Delta_{p_0+M-1,n}=\Delta_{p_0,n}=O(n^{-\gamma_0}),\  \eta_n=O(n^{-\gamma_1}),\\ V_{1n}=O(n^{\gamma_2}),\ V_{2n}=O(n^{\gamma_3}),\ C_{1n}=O(n^{\gamma_4}),\ C_{2n}=O(n^{\gamma_5}),\ C_{3n}
		=O(n^{\gamma_6}),\ C_{4n}=O(n^{\gamma_7}),\ p_0=O(n^{\gamma_8}),\ M=O(n^{\gamma_9})$, where
		$\gamma_1,\gamma_8,\gamma_9>0,\ \gamma_8<\gamma_9,\ 3\gamma_9<\gamma_0,\ \gamma_9<2\gamma_1,\ \gamma_7<\gamma_1,\ \gamma_6<\gamma_1,\ \gamma_2<\gamma_0,\ \gamma_3<\gamma_0,\ 3\gamma_9+\gamma_4<\gamma_0,\ 3\gamma_9<\gamma_1,\ 6\gamma_9+\gamma_5<\gamma_0,\ 3\gamma_8+\gamma_4+\gamma_6<\gamma_0,\ 3\gamma_8+\gamma_6<\gamma_1,\ 6\gamma_8+\gamma_5+\gamma_6<\gamma_0,\
		3\gamma_9+\gamma_4+\gamma_7<\gamma_0,\ 3\gamma_9+\gamma_7<\gamma_1,\ 6\gamma_9+\gamma_5+\gamma_7<\gamma_0$. Then Assumptions \ref{assumption 2.4} and \ref{assumption 2.5} hold.
	\end{theorem}
	
	Furthermore, note that the inputs of computer models can be designed. We next show that, for fixed $p_0$ and $M$, if the inputs are generated by simple random sampling, then the two assumptions, and thus Theorem \ref{THM4}, hold with a probability tending to one.

	\begin{assumption}\label{assumption rs} The design matrix $\m{X}$ is generated by simple random sampling, i.e., $x_{ij},\ i=1,\ldots,n,\ j=1,\ldots, p$, are independently identically distributed from uniform distribution on $[0,1)$.
	\end{assumption}

	\begin{assumption}\label{assumption gamma}
		Let $p_0$ and $M$ be fixed with $p_0\leqslant M$. As $n\to \infty$, $\eta_n=O(n^{-\gamma_1}),\ V_{1n}=O(n^{\gamma_2}),\ V_{2n}=O(n^{\gamma_3}),\ C_{1n}=O(n^{\gamma_4}),\ C_{2n}=O(n^{\gamma_5}),\ C_{3n}
		=O(n^{\gamma_6}),\ C_{4n}=O(n^{\gamma_7}),\ \log p=O(n^{\gamma_{10}})$, where
		$\gamma_1,\gamma_{10}>0,\ \gamma_{10}<1, \  \gamma_2<\tilde{\gamma}_0=(1-\gamma_{10})/2,\ \gamma_3<\tilde{\gamma}_0,\ \gamma_4<\tilde{\gamma}_0,\
		\gamma_5<\tilde{\gamma}_0,\
		\gamma_7<\gamma_1,\ \gamma_6<\gamma_1,\ \gamma_4+\gamma_6<\tilde{\gamma}_0,\ \gamma_5+\gamma_6<\tilde{\gamma}_0,\ \gamma_4+\gamma_7<\tilde{\gamma}_0,\ \gamma_5+\gamma_7<\tilde{\gamma}_0 $.
	\end{assumption}
	
	Such $\gamma_1,\ldots,\gamma_7,\gamma_{10}$ in Assumption \ref{assumption gamma} exist. For example, take $\gamma_1=0.5,\gamma_2=\gamma_3=0.2,\gamma_4=\gamma_5=0.15,\gamma_6=\gamma_7=0.05,\gamma_{10}=0.5$.
	
	\begin{theorem}\label{THMpr}
		Under Assumptions \ref{assumption 2.1}, \ref{assumption 2.2}, \ref{assumption rs} and \ref{assumption gamma}, as $n\to\infty$,\begin{align*}P\left(\max\limits_{\s{A}\in\s{U}_0}\left\|\v{y}-\vg{1}_n\widehat{\beta}_0(\s{A})-\m{X}\widehat{\vg{\beta}}_{Z_p}({\s{A}})\right\|^2
		<\min\limits_{\s{A}_1\in\s{U}_1}\left\|\v{y}-\vg{1}_n\widehat{\beta}_0(\s{A}_1)-\m{X}\widehat{\vg{\beta}}_{Z_p}({\s{A}_1})\right\|^2\right)\to1.
		\end{align*}
	\end{theorem}

Note that screening can be viewed as a step of data preprocessing for high-dimensional data. One should conduct further steps of variable selection or sensitivity analysis to remove redundant variables after screening $M$ variables (Fan and Lv 2008). Therefore, the selection of $M$ is not very crucial. Fan and Lv (2008) suggested $M=n/\log(n)$. Cross-validation methods can also be used to specify $M$.

	\section{Use of other basis functions}\label{sec:db}
		
	The key point why linear screening methods are valid for nonlinear model \eqref{f} is \eqref{3} in Assumption	\ref{assumption 2.1}, which guarantees that the active variables of \eqref{f} are still active in its BLA. If this assumption does not hold, i.e., for some $j\in Z_{p_0}$,
	\begin{align}\label{e0}\int_{[0,1]^p}x_jf(\v{x})d\v{x}-\frac{1}{2}\int_{[0,1]^p} f(\v{x})d\v{x}=0,
	\end{align}  then we cannot select the $j$th active variable by linear screening. In fact, the possibility of the extreme case \eqref{e0} is usually negligible in practice. In modeling for computer experiments, $f$ is usually assumed to be a realization from a Gaussian process (Santner, Williams, and Notz 2018), and thus the probability that \eqref{e0} occurs is zero. In general, \eqref{e0} occurs only for some artificial functions. For example, in one dimension, \eqref{e0} occurs for $f(x)=10(x-{1}/{2})^2$; see Figure 1. It will be shown from our numerical results in Section \ref{Sec:numerical_exp} that the proposed linear screening methods perform quite well for most practical cases. Even so, we now present methods to handle the extreme cases when \eqref{e0} occurs.
	
	\begin{figure}[t]\label{fig:compare}
		\scalebox{0.6}[0.6]{\includegraphics{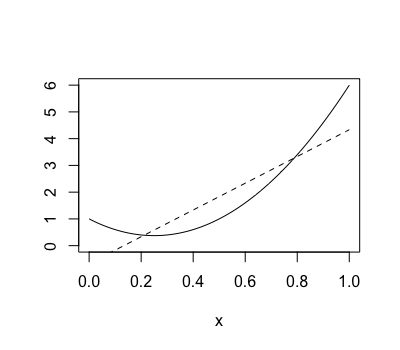}}
		\scalebox{0.6}[0.6]{\includegraphics{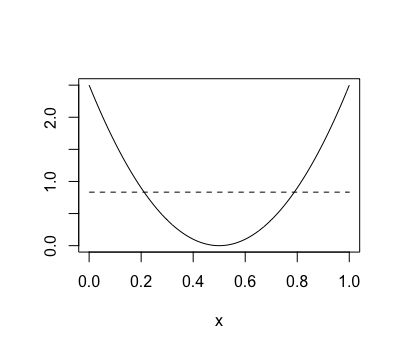}}
		\caption{Solid lines and dashed lines represent the original functions and the corresponding BLA's, respectively. On the left-hand side and right-hand side, $f$ is taken as $10x^2-5x+1$ and $10(x-{1}/{2})^2$, respectively, and the latter yields a case where \eqref{e0} occurs.}
	\end{figure}
	
	\begin{figure}[t]\label{fig:exm}
		\begin{center}
			\scalebox{0.6}[0.6]{\includegraphics{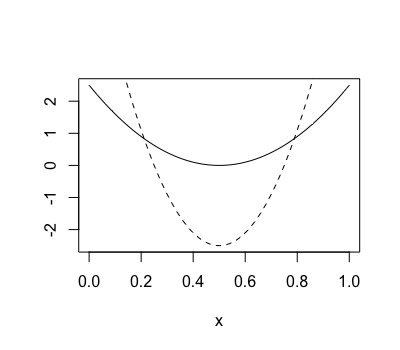}}
		\end{center}
		\caption{The solid line and dashed line represent the original function $f(x)=10(x-{1}/{2})^2$ and its general BLA with $b(x)=-4x^2+4x-{2}/{3}$, respectively. }
	\end{figure}
	
	Consider a general form of \eqref{lm},\begin{equation}\label{glm}y=\phi_0+\vg{\phi}'\v{b}(\v{x})+\epsilon,\end{equation}where $\v{b}(\v{x})=\left(b(x_1),\ldots,b(x_p)\right)'$ are pre-specified basis functions and $b\in C[0,1]$. It is clear that \eqref{lm} is a special case when taking $b(x)=x$. Similarly, the general BLA of $f$ based on the basis function $\v{b}(\v{x})$ is defined as\begin{align*}
	\beta_0+\vg{\beta}'\v{b}(\v{x})
	=\mathop{\arg\min}_{g\in \{\phi_0+{\vg{\phi}}'\v{b}(\v{x}):\ \phi_0 \in \mathbb{R},\ \vg{\phi} \in \mathbb{R}^p\}}\int_{[0,1]^p}\left[f(\v{x})-g(\v{x})\right]^2d\v{x},
	\end{align*}and the corresponding linear screening methods can be established by using \eqref{glm} to model the data from \eqref{f}. Here the general linear screening methods do not work only when for some $j\in Z_{p_0}$,
	\begin{align}\label{ge0}\int_{[0,1]^p}b(x_j)f(\v{x})d\v{x}-\int_{[0,1]} b(x_j)d{x_j}\int_{[0,1]^p} f(\v{x})d\v{x}=0,
	\end{align}	Therefore, when \eqref{e0} occurs, the linear screening method can still be valid with different $b$ that avoids the occurrence of \eqref{ge0}. For example, for the above function $f(x)=10(x-{1}/{2})^2$ that leads to \eqref{e0}, we can use the quadratic basis function, \begin{equation}\label{b}b(x)=-4x^2+4x-{2}/{3},\end{equation} which is orthogonal to $b(x)=x$, and the general linear screening method is valid with this basis; see Figure 2.

From the above discussion, we present a two-stage strategy to improve the credibility of screening results. In the two stages, we use the linear basis $b(x)=x$ and the quadratic basis \eqref{b} in our (general) linear screening method, respectively, and then combine the results from the two stages.

Here we give further discussions on the selection of basis functions. First, we take a simple quadratic function for example to compare the linear and quadratic basis functions.
Suppose the quadratic function is $$f(x_1,x_2)=\beta_0+\beta_1x_1+\beta_2x_2+\beta_3x_1^2+\beta_4x_2^2+\beta_5x_1x_2,\ (x_1,x_2)'\in[0,1]^2.$$
We calculate the coefficients of $x_1$ and $x_2$ in the (general) BLAs with respect to the the linear and quadratic basis functions, respectively, and the corresponding approximation losses, where the approximation loss is defined as
\begin{align*}
L(f) = \int_{[0,1]^2}\left[f(\v{x})-g(\v{x})\right]^2d\v{x}
\end{align*}
and $g(\v{x})$ denotes the corresponding (general) BLA. The results are displayed in Table \ref{tab:loss}. Suppose only $x_1$ is active in $f$, i.e., $\beta_2=\beta_4=\beta_5=0$. We can see that there is at least one basis function that can select the active variable since the two equations, $\beta_1+\beta_3=0$ and $15\beta_1/16+\beta_3=0$, cannot simultaneously hold (otherwise $x_1$ is not active). In fact, when $\beta_1$ and $\beta_3$ are continuous random variables, say, normal random variables, either of the two basis functions can select the active variable with probability one. On the other hand, the coefficients of $x_1$ from the two basis functions are similar, and that the approximation loss of the linear basis is even better in many cases, depending on the values of $\beta_1$ and $\beta_3$. This indicates that the linear basis seems comparable to the quadratic basis even for quadratic functions in our linear approximation method. This is one reason why the simple linear basis should be tried at first.

\begin{table}[t]
	\centering\footnotesize
	\footnotesize
	\caption{Comparison of the two basis functions for the quadratic function}\label{tab:loss}
	\begin{tabular}{lcc}
		\hline \hline
		
		&  linear basis& quadratic basis  \\
		\hline
		coefficient of $x_1$ &$\beta_1+\beta_3+\beta_5/2$  &$15\beta_1/16+\beta_3+15\beta_5/32 $  \\
		coefficient of $x_2$ & $ \beta_2+\beta_4+\beta_5/2$ & $15\beta_2/16+\beta_4+15\beta_5/32  $  \\
		$L(f)$ & $\beta_3^2/180+\beta_4^2/180+\beta_5^2/144$ & $\beta_1^2/192+\beta_2^2/192 +11\beta_5^2/1152+\beta_1\beta_5/192+\beta_2\beta_5/192$  \\
		\hline
	\end{tabular}\\[3mm]
\end{table}

Second, from Table \ref{tab:loss} we can see that, when $x_1$ and $x_2$ are both active, the linear and quadratic basis functions both fail for $\beta_0=1/4,\ \beta_1=\beta_2=-1/2,\ \beta_3=\beta_4=0$, and $\beta_5=1$, which corresponds to $f(x_1,x_2)=(x_1-1/2)(x_2-1/2)$. In fact, none of basis functions works for this function. In other words, when using the nonparametric additive model (Ruppert, Wand, and Carroll 2003), $$y=\phi_0+b_1(x_1)+\cdots+b_p(x_p)+\epsilon,$$to substitute \eqref{glm}, where $b_1,\ldots,b_p\in C[0,1]$ are unspecified basis functions, the corresponding screening method can still be invalid for some $f$; the proof can be found in Section D of the Supplementary Materials. Off course such functions are much rarer than those satisfying \eqref{ge0} for a specific basis function.

Third, by Theorem \ref{THM1}, it is desirable for selecting the active variable that the absolute value of the left part in \eqref{ge0} is as large as possible. This inspires us to consider the selection of the optimal $b$.
Without loss of generality, let $\int_{[0,1]^p} f(\v{x})d\v{x}=0$, and then the optimal $b$ should maximize $\left|\int_{[0,1]^p}b(x_j)f(\v{x})d\v{x}\right|$ for $j=1,\ldots,p$. It can be seen that we cannot specify the optimal basis since it depends on the unknown model $f$. Overall, it seems impossible to select a ``perfect" basis function. Relatively simple linear and quadratic basis functions are satisfactory for our problem.

	\section{Numerical experiments}\label{Sec:numerical_exp}
	
	\subsection{Three test functions}\label{subsection:3f}
		
	This subsection considers the following test functions on $[0,1]^p$,	
	\begin{eqnarray*}&&\mathrm{(I)}\ f(\v{x})=\sum\limits_{j=1}^p jx_j^2,\\
		&&\mathrm{(II)}\ f(\v{x})=-20\exp \left(-\dfrac{1}{5}\sqrt{\dfrac{1}{p}\sum\limits_{j=1}^{p}x_j^2}\right)-\exp \left	 (\dfrac{1}{p}\sum\limits_{j=1}^p 2 \pi x_j\right)+20+\exp(1),
		\\&&\mathrm{(III)}\ f(\v{x})=\left(\sum\limits_{j=1}^px_j\right)\exp\left[-\sum\limits_{j=1}^p\sin\left(x_j^2\right)\right].\end{eqnarray*}
	Model $\mathrm{(I)}$ is known as weighted sphere model, model $\mathrm{(II)}$ is Ackley's model, and model $\mathrm{(III)}$ is Yang's model (Yang 2010). The design matrix $\m{X}$ is generated by simple random sampling. Three combinations of $(n,p,M)$ and two values of $p_0$ are considered; see Table \ref{tab:f}.
	
	Three screening methods for linear models are used in the proposed linear screening method: Fan and Lv (2008)'s sure independence screening, Tibshirani (1996)'s lasso, and Xiong (2014)'s fast orthogonalizing subset screening, which are denoted by L-SIS, L-Lasso, and L-FOSS, respectively. We use cross-validation to specify the tuning parameter in L-Lasso, and only keep the variables with the largest $M$ absolute values of coefficients if the number of selected variables is larger than $M$. The initial point in L-FOSS is taken as the L-Lasso solution. Two model-free screening methods, Zhu et al. (2011)'s sure independent ranking and screening (SIRS) and Li, Zhong, and Zhu (2012)'s sure independence screening procedure based on the distance correlation (DC-SIS), are compared with our linear screening methods.
	
	The coverage rates that the selected subset include the true submodel over 1000 repetitions are given in Table \ref{tab:f}. It can be seen from the table that the linear screening methods have better overall performance than the two model-free marginal screening methods. In particular, L-FOSS performs the best among all the methods since it is an effective algorithm for solving the $\ell_0$ problem \eqref{bsr}. An interesting finding is that, even for nonlinear models, the linear marginal L-SIS is better than the two model-free marginal methods for many cases.
	
	\begin{table}[t]
		\centering\footnotesize
		\footnotesize
		\caption{Coverage rates in Section \ref{subsection:3f}}\label{tab:f}
		\begin{tabular}{lcccccc }
			\hline \hline
			\multicolumn{7}{c}{function (\uppercase\expandafter{\romannumeral 1})}\\
			\hline
			&\multicolumn{2}{c}{$n=100,p=200,M=30$} &\multicolumn{2}{c}{$n=200,p=500,M=50$}&	\multicolumn{2}{c}{$n=100,p=1000,M=50$} \\
			\hline
			&  $p_0=5$& $p_0=10$  &  $p_0=5$& $p_0=10$ &   $p_0=5$& $p_0=10$ \\
			\hline
			SIRS & 0.297 &0.007  & 0.428&0.026& 0.136&0.000  \\
			DC-SIS & 0.307 & 0.008 & 0.492 & 0.026  & 0.149 &0.001   \\
			L-SIS & 0.355 & 0.014  & 0.543& 0.030  & 0.177 & 0.001 \\
			L-Lasso & 0.953 & 0.299 & 1.000& 0.481 & 0.846 & 0.056 \\
			L-FOSS & 0.988& 0.337  & 1.000& 0.597 & 0.897& 0.064 \\
			\hline
		\end{tabular}\\[3mm]
		
		\begin{tabular}{lcccccc }
			\hline \hline
			\multicolumn{7}{c}{function (\uppercase\expandafter{\romannumeral 2})}\\
			\hline
			&\multicolumn{2}{c}{$n=100,p=200,M=30$} &\multicolumn{2}{c}{$n=200,p=500,M=50$}&	\multicolumn{2}{c}{$n=100,p=1000,M=50$}\\
			\hline
			&  $p_0=5$& $p_0=10$  &  $p_0=5$& $p_0=10$ &  $p_0=5$& $p_0=10$ \\
			\hline
			SIRS & 0.941 &0.163  & 1.000 &0.777& 0.828 &0.022\\
			DC-SIS & 0.981 & 0.263  & 1.000 & 0.885 & 0.912 & 0.047\\
			L-SIS & 0.957 & 0.250  & 0.999& 0.863& 0.876& 0.046 \\
			L-Lasso & 0.982 & 0.705 & 1.000 & 0.997 & 0.908& 0.144 \\
			L-FOSS & 0.998& 0.808 & 1.000 & 1.000  & 0.986& 0.199\\
			\hline
		\end{tabular}\\[3mm]
		\begin{tabular}{lcccccc }
			\hline \hline
			\multicolumn{7}{c}{function (\uppercase\expandafter{\romannumeral 3})}\\
			\hline
			&\multicolumn{2}{c}{$n=100,p=200,M=30$} &\multicolumn{2}{c}{$n=200,p=500,M=50$}&	\multicolumn{2}{c}{$n=100,p=1000,M=50$}\\
			\hline
			&  $p_0=5$& $p_0=10$  &  $p_0=5$& $p_0=10$ &  $p_0=5$& $p_0=10$ \\
			\hline
			SIRS & 0.971 &0.366 & 1.000&0.936 & 0.902 &0.101 \\
			DC-SIS & 0.987 & 0.401 & 1.000& 0.944& 0.943 & 0.115\\
			L-SIS & 0.987 & 0.424 & 1.000  & 0.943  & 0.942 & 0.136\\
			L-Lasso & 0.997 & 0.978 & 1.000 & 1.000 & 0.994& 0.562  \\
			L-FOSS &1.000& 0.995 & 1.000 & 1.000  &0.999& 0.644\\
			\hline
		\end{tabular}\\[3mm]
	\end{table}

	\subsection{Borehole model}

	The following borehole model (Worley 1987)
	\begin{equation}\label{bh}
	y=\frac{2\pi T_u(H_u-H_l)}{\displaystyle \log(r/r_w)\Big[1+\frac{\displaystyle
			2LT_u}{\displaystyle \log(r/r_w)r_w^2K_w}+T_u/T_l\Big]}
	\end{equation} that describes the flow rate through a borehole is widely used in computer experiments (Morris, Mitchell, and Ylvisaker 1993; Xiong, Qian, and Wu 2013). The ranges of the eight input variables in this model are $r_\omega \in [0.05,0.15]\,{\mathrm{m}},\ r \in [100,50000]\,{\mathrm{m}},\ T_u \in [63070,115600]\,{\mathrm{m^2/yr}},\\ H_u \in [990,1110]\,{\mathrm{m}},\ T_l \in [63.1,116]\,{\mathrm{m^2/yr}},\ H_l \in [700,820]\,{\mathrm{m}},\ L\in [1120,1680]\,{\mathrm{m}},$ and $K_\omega \in [1500,15000]\,{\mathrm{m/yr}}$. We augment the dimension of the borehole model to $p=100$ and 500 by adding noisy input variables, and consider two combinations of $(n,p,M)$; see Table 3. The design matrix is generated by simple random sampling in the simulation. The five methods in Section \ref{subsection:3f} are compared.
	
	The borehole model itself is sparse. The Sobol' indices (Sobol' and Saltelli 1995) of the eight input variables in \eqref{bh} are $0.5713,\ <5\times10^{-5},\ <5\times10^{-5},\ 0.0356,\ <5\times10^{-5},\ 0.0357,\ 0.0342,\ 0.4649$, respectively.
	First, we only consider the first and eighth variables as active variables, and compute the coverage rates that the selected subset includes the two variables of the five methods over 1000 repetitions. For the two cases of $(n,p,M)=(50,100,30)$ and $(n,p,M)=(200,500,30)$, all the methods can correctly screen the two variables over all the repetitions. Second, we add the fourth, sixth, and seventh variables in the set of active variables, and compute the coverage rates that the selected subset includes the five variables. The simulation results are shown in Table \ref{tab:bh}. We can see that, it is difficult for SIRS, DC-SIS, and L-SIS to screen the active variables, and L-FOSS performs much better than them.

	\begin{table}[t]
		\centering\footnotesize
		\footnotesize
		\caption{Coverage rates of the borehole model with five active variables}\label{tab:bh}
		\begin{tabular}{lcc}
			\hline \hline
			&\multicolumn{1}{c}{$n=50,p=100,M=30$} &\multicolumn{1}{c}{$n=200,p=500,M=30$}	\\
			\hline
			SIRS & 0.088 &0.096\\
			DC-SIS & 0.094 & 0.120\\
			L-SIS & 0.111 & 0.198 \\
			L-Lasso & 0.108& 0.838\\
			L-FOSS & 0.282 & 0.980\\
			\hline
		\end{tabular}\\[2mm]
	\end{table}
	
In addition, we conduct a small simulation to evaluate our L-FOSS method with a data-driven $M$. Let $M_0$ denote the number of active variables from L-Lasso. Note that Fan and Lv (2008) suggested $M=n/\log(n)$. We select $M\in[\min\{M_0,n/\log(n)\},\max\{M_0,n/\log(n)\}]$ by minimizing the generalized cross-validation (GCV) criterion (Golub, Heath, and Wahba 1979),
$$\mathrm{GCV}(M)=\frac{\min_{\s{A}\subset Z_n,|\s{A}|=M}\left\|\v{y}-\left(\widehat{\beta}_0(\s{A})\v{1}_n+\m{X}_{\s{A}}\widehat{\vg{\beta}}(\s{A})\right)\right\|^2}{n(1-M/n)^2},$$where $\v{1}_n$ denotes the $n$-vector $(1,\ldots,1)'$ and $\left(\widehat{\beta}_0(\s{A}),\widehat{\vg{\beta}}(\s{A})'\right)'$ are the least squares estimators in \eqref{1}. We consider the case of $(n,p)=(200,500)$ and implement L-FOSS with such $M$. The mean and standard deviation of the selected $M$ over 1000 repetitions are 36.005 and 1.398, respectively. The coverage rate over the 1000 repetitions is 0.982. We also compute the coverage rates of correctly selecting individual active variables (the first, fourth, sixth, seventh, eighth variables), which are $1.000,\ 0.988,\ 0.996,\ 0.993$, and $1.000$, respectively. The high rates of correctly selecting the first and eighth variables reflect the fact mentioned before that they have the largest values of the sensitivity index.

	\subsection{Quadratic basis}
	\label{subsec:t}
	
	It can be seen from the previous subsections that the linear screening method such as F-FOSS with the linear basis function $b(x)=x$ in \eqref{glm} performs quite well for various cases. For some extreme cases where the linear basis does not work, Section \ref{sec:db} points out that linear screening can still be valid with different basis functions. In this subsection we conduct a small simulation to verify this point.
	
	Here we consider the one-dimensional function $f(x)=10(x-{1}/{2})^2$ in Section \ref{sec:db}, which leads to invalidness of linear basis, and augment the dimension by adding noisy input variables. The design matrix is generated by simple random sampling in the simulation. We consider the two linear screening methods, L-Lasso and L-FOSS, and use the linear basis function and the quadratic basis function \eqref{b} in them. The two-stage method proposed in Section \ref{sec:db} is also compared. It compares the residual sums of squares of the selected subset from the two basis functions, and use the subset corresponding to the smaller as the final result.
The coverage rates of these methods are shown in Table \ref{tab:db}. We can see that the linear basis does indeed yield bad screening results, and that the quadratic basis function and the two-stage method improve it obviously.
	
	\begin{table}[t]
		\centering\footnotesize
		\footnotesize
		\caption{Coverage rates in Section \ref{subsec:t}}\label{tab:db}
		\begin{tabular} {lcccc}
			\hline \hline
			\multicolumn{5}{c}{$n=50,p=100,M=5,p_0=1$} \\
			\hline
			&\quad&L-Lasso&L-FOSS\\
			\hline
			linear basis&\quad&0.008 & 0.110\\
			quadratic basis&\quad& 1.000&1.000\\
			two-stage&\quad& 1.000&1.000\\\hline
		\end{tabular}\\
	\end{table}

	\section{Discussion}\label{Sec:summary}

	In this paper we have developed linear screening methods to screen active input variables for high-dimensional computer experiments. Numerical investigations show that the proposed methods are very effective. In particular, L-FOSS performs much better than existing model-free marginal screening methods. Our methods can be viewed as analogues of the linear model method in sensitivity analysis, and we have provided theoretical guarantees for them based on the theory of BLA.


\section*{Appendix}
	\setcounter{equation}{0}
	\renewcommand{\thesubsection}{\Alph{subsection}}
	\renewcommand{\theequation}{A.\arabic{equation}}
	\subsection{Hardy-Kruse variation and $L_\infty$ discrepancy}
	\begin{definition} (Owen 2005)\label{def:v}
		Let $g:[0,1]^p \to \mathbb{R}$. If $J$ is a sub-rectangle of $[0,1]^p,$ let $\Delta_J(g)$ be the sum of the values of $g$ at the $2^p$ vertices of $J$, with alternating signs at nearest neighbour vertices. The Vitali variation of $g:[0,1]^p \to \mathbb{R}$ is defined to be
		\begin{align*}
		V^{Vit}(g):=\sup\left\{
		\sum_{J\in \Pi}|\Delta_J(g)| \left |
		\begin{gathered}
		\Pi  \text{\ is a partition of }[0,1]^p \text{\ into finitely}\\  \text{many non-overlapping sub-rectangles.}
		\end{gathered}
		\right\} \right .
		\end{align*}
		For $1\leqslant s\leqslant p$, the Hardy-Krause variation of $g$ is defined to be
		\begin{align*}
		V_{HK}(g):=\sum_F V^{Vit}(g|G),
		\end{align*}
		where the sum runs over all faces $G$ of $[0,1]^p$ having dimension at most $s$.
	\end{definition}
	\begin{definition}\label{def:d}	(Heinrich et al. 2001)
		Let $\v{x}_i=(x_{i1},\cdots,x_{ip})'\in[0,1]^p,\ i=1,\cdots,n$ and $F_n(\v{x})$ be the empirical distribution of the points for $\v{x}=(x_1,\ldots,x_p)'$.
		The $L_\infty$ discrepancy $\nu_{p,n}$ of $\{\v{x}_1,\ldots,\v{x}_n\}$ is defined as
		\begin{align}\label{nu}
		\nu_{p,n}=\sup_{\v{x} \in [0,1]^p}\left|F_n(\v{x})-\prod_{i=1}^p x_i\right|
		\end{align}
	\end{definition}

	\subsection{Lemmas}
	\begin{lemma} (Chen and Li 2003)\label{lemma1}
		For two $n\times n$ real symmetric matrices $\m{A}$ and $\m{B}$, let the eigenvalues of $\m{A}$ and $\m{B}$ be $\lambda_1\geqslant\cdots\geqslant\lambda_n$ and $\mu_1\geqslant\cdots\geqslant\mu_n$, respectively. Then for any $i=1,\ldots,n$,
		\begin{align*}
		|\lambda_i-\mu_i|\leqslant\|\m{B}-\m{A}\|.
		\end{align*}
	\end{lemma}
	\begin{lemma}\label{lemma2}
		{\bf(Koksma-Hlawka inequality)} (Tezuka 2002) If $g\in$ BVHK on $[0,1]^p$, then for any $\v{x}_1,\cdots,\v{x}_n \in[0,1)^p$,
		\begin{align*}
		\left|\int_{[0,1]^p}g(\v{x})d\v{x}-\dfrac{1}{n}\sum\limits_{n=1}^{n}g(\v{x}_n)\right| \leqslant \nu_{k,n} V_{HK}(g),
		\end{align*}
		where $V_{HK}(g)$ and $\nu_{k,n}$ are defined in Definitions \ref{def:v} and \ref{def:d}, respectively.
	\end{lemma}
	
	\begin{lemma} \label{lemma3}	(Owen 2005)
		For a set $u \subset Z_p$, let $u^c=Z_p\backslash u$ denote its complement. Let $\partial^u g$ denote the partial derivative of $g$ taken once with respect to each variable $j\in u$. For $\v{x} \in [0,1]^p$ and $u \subset Z_p$ let $\v{x}_u:\v{1}_{u^c}$ be the point $\v{y} \in [0,1]^p$ with $y_j=x_j$ for $j\in u$ and $y_j=1$ for $j \in u^c$. If the mixed partial derivative $\partial^{1:p}g$ exists, then
		\begin{align*}
		V_{HK}(g)\leqslant \sum_{u\neq \emptyset}\int_{[0,1]^{|u|}}\left|\partial^ug(\v{x}_u:\vg{1}_{u^c})\right|d\v{x}_u.
		\end{align*}
	\end{lemma}
	
	\begin{lemma} (Owen 2005)\label{lemma:owen2}
		Let $f$ and $g$ be functions on $[0,1]^p$. If $f,g \in$ BVHK ,  then $f+g,f-g,fg\in$ BVHK.
	\end{lemma}
	
	\begin{lemma}\label{lemma4}
		(Kiefer and Wolfowitz 1958)
		Let $F$ be a distribution function on Euclicean $m$-space and $\v{x}_1,\cdots,\v{x}_n$ be independent chance variables with distribution function $F$. $F_n$ is empirical distribution function of the points. For each $m$, there exists positive constants $c_0$ and $c$, such that for all $n$, all $F$ and all positive $r$,
		$$P\left(\sup\limits_{\v{x}} |F(\v{x})-F_n(\v{x})|\leqslant r\right)\geqslant 1-c_0\exp(-cnr^2).$$
	\end{lemma}
	\begin{lemma}\label{lemma6}
		Under Assumption 1, for $\s{A} \subset Z_p$, we have
		\begin{align}
		\beta_{Z_p}(\s{A})_j&=12\left(\int_{[0,1]^p}x_jf_n(\v{x})d\v{x}-\frac{1}{2}\int_{[0,1]^p}f_n(\v{x})d\v{x}\right),\ j \in \s{A},\label{betaAj}\\
		\beta_0(\s{A})&=\int_{[0,1]^p}f_n(\v{x})d\v{x}-\sum_{j \in \s{A}}\beta_{Z_p}(\s{A})_j/2.\label{betaA0}
		\end{align}
		Furthermore,
		\begin{align}
		|\beta_{Z_p}(\s{A})_j|&>12\tau, \quad j \in \s{A}\cap \s{A}_0,\label{eb1}\\
		|\beta_{Z_p}(\s{A})_j|&<12\eta_n, \quad j \in \s{A} \setminus{A}_0.\label{eb2}
		\end{align}
		\begin{proof}
			Let $\s{A}=\{a_1,\cdots,a_m\}$. By (7), let
			\begin{align*}
			\frac{\partial G}{\partial \phi_0}&=\int_{[0,1]^p}(f_n(\v{x})-(\phi_0+\phi_{a_1}x_{a_1}+\cdots+\phi_{a_m}x_{a_m}))d\v{x}=0, \\
			\frac{\partial G}{\partial \phi_{a_j}}&=\int_{[0,1]^p}(f_n(\v{x})-(\phi_0+\phi_{a_1}x_{a_1}+\cdots+\phi_{a_m}x_{a_m}))x_{a_j}d\v{x}=0,\ j=1,\cdots,m,
			\end{align*}
			which lead to
			\begin{align}\label{betaZw}
			\begin{pmatrix}
			\beta_0({\s{A}})\\\vg{\beta}({\s{A}})
			\end{pmatrix}&=
			\begin{pmatrix}
			1&&\int_{[0,1]^p}x_{a_1}d\v{x}&&\cdots&&\int_{[0,1]^p}x_{a_m}d\v{x}\\
			\int_{[0,1]^p}x_{a_1}d\v{x}&&\int_{[0,1]^p}x_{a_1}^2d\v{x}&& \cdots&&\int_{[0,1]^p}x_{a_1}x_{a_m}d\v{x}\\
			\vdots&&\vdots&&\ddots&&\vdots\\
			\int_{[0,1]^p}x_{a_m}d\v{x}&&\int_{[0,1]^p}x_1x_{a_m}d\v{x}&&\cdots&&\int_{[0,1]^p}x_{a_m}^2d\v{x}
			\end{pmatrix}^{-1}
			\begin{pmatrix}
			\int_{[0,1]^p}f_n(\v{x})d\v{x}\\
			\int_{[0,1]^p}x_{a_1}f_n(\v{x})d\v{x}\\
			\vdots\\
			\int_{[0,1]^p}x_{a_m}f_n(\v{x})d\v{x}
			\end{pmatrix}\notag \\&=\m{U}^{-1}\v{v},
			\end{align}
			where\begin{align} \label{Zw}
			\m{U}=\begin{pmatrix}
			1&& \frac{1}{2}&&  \ldots&& \frac{1}{2}\\
			\frac{1}{2}&&\frac{1}{3}&& \ldots&& \frac{1}{4}\\
			\vdots&&\vdots&&\ddots&&\vdots\\
			\frac{1}{2}&&\frac{1}{4}&&  \ldots&& \frac{1}{3}
			\end{pmatrix}_{(m+1)\times (m+1)},
			\v{v}=\begin{pmatrix}
			\int_{[0,1]^p}f_n(\v{x})d\v{x}\\
			\int_{[0,1]^p}x_{a_1}f_n(\v{x})d\v{x}\\
			\vdots\\
			\int_{[0,1]^p}x_{a_m}f_n(\v{x})d\v{x}
			\end{pmatrix}.
			\end{align}
			Some algebra yields
			\begin{align} \label{Zn}
			\m{U}^{-1}=12\begin{pmatrix}
			\frac{1+3m}{12}&-\frac{1}{2}&\cdots&-\frac{1}{2}\\
			-\frac{1}{2}&1&\cdots&0\\
			\vdots&\vdots&\ddots&\vdots\\
			-\frac{1}{2}&0&\cdots&1
			\end{pmatrix}.
			\end{align}
			By \eqref{betaZw} and \eqref{Zn}, we get
			(\ref{betaAj}) and (\ref{betaA0}).
			
			By (6), we get
			\begin{align*}
			|\beta_{Z_p}(\s{A})_j|=12\left |\int_{[0,1]^p}x_jf_n(\v{x})d\v{x}-\frac{1}{2}\int_{[0,1]^p}f_n(\v{x})d\v{x}\right|>12\tau, \quad j \in \s{A} \cap \s{A}_0.
			\end{align*}Furthermore, let \begin{equation}r_n(\v{x})=f_n(\v{x})-\widetilde{f}_n(\v{x}_{\s{A}_0}).\label{r}\end{equation} By (5), for $j \in
			\s{A}\setminus \s{A}_0$,
			\begin{align*}
			|\beta_{Z_p}(\s{A})_j|&=12\left |\int_{[0,1]^p}x_jf_n(\v{x})d\v{x}-\frac{1}{2}\int_{[0,1]^p}f_n(\v{x})d\v{x}\right|\\
			&=12\left |\int_{[0,1]^p}x_jr_n(\v{x})d\v{x}-\frac{1}{2}\int_{[0,1]^p}r_n(\v{x})d\v{x}\right| \\
			&\leqslant 12\left(\left |\int_{[0,1]^p}x_jr_n(\v{x})d\v{x}\right|+\left|\frac{1}{2}\int_{[0,1]^p}r_n(\v{x})d\v{x}\right|\right)\\
			&< 12 \left(\eta_n \int_{[0,1]^p}x_jd\v{x}+\frac{1}{2}\eta_n\right)=12\eta_n.
			\end{align*}
			This completes the proof.
		\end{proof}
	\end{lemma}
	\begin{lemma}\label{lemma:a0}
		Under Assumption 1, for $\s{A}_1 \in \s{U}_1$, we have
		\begin{align*}
		&\int_{[0,1]^p}\left(f_n(\v{x})-\beta_0(\s{A}_1)- \vg{\beta}_{Z_p}(\s{A}_1)'\v{x}\right)^2d\v{x}-\int_{[0,1]^p}\left(f_n(\v{x})-\beta_0(\s{A}_0)-\vg{\beta}_{Z_p}(\s{A}_0)'\v{x}\right)^2d\v{x}  \\&\geqslant 12\tau^2-12(M-p_0+1)\eta_n^2,	 \end{align*}
		where $\beta_0(\s{A}_1),\vg{\beta}_{Z_p}(\s{A}_1),\beta_0(\s{A}_0),\vg{\beta}_{Z_p}(\s{A}_0)$ are defined by (7) and (8).
	\end{lemma}
	\begin{proof}	Without loss of generality, assume $\s{A}_1=\{1,\ldots,d,p_0+1,\ldots,p_0+M-d\}$.
		By Lemma \ref{lemma6}, $\beta_{Z_p}(\s{A}_1)_j=\beta_{Z_p}(\s{A}_0)_j,\ j=1,\cdots,d$. Therefore, we have
		\begin{align*}
		&\int_{[0,1]^p}(f_n(\v{x})-\beta_0(\s{A}_1)- \vg{\beta}_{Z_p}(\s{A}_1)'\v{x})^2d\v{x}-\int_{[0,1]^p}(f_n(\v{x})-\beta_0(\s{A}_0)-\vg{\beta}_{Z_p}(\s{A}_0)'\v{x})^2d\v{x}  \\
		&=\int_{[0,1]^p}\left(f_n(\v{x})-\sum\limits_{j=1}^dx_j\beta_{Z_p}(\s{A}_0)_j-\sum\limits_{j=p_0+1}^{p_0+M-d}x_j\beta_{Z_p}(\s{A}_1)_j-\beta_0({\s{A}_1})\right)^2d\v{x}
		\\&\quad\quad-\int_{[0,1]^p}\left(f_n(\v{x})-\sum\limits_{j=1}^{p_0}x_j\beta_{Z_p}(\s{A}_0)_j-\beta_0({\s{A}_0})\right)^2d\v{x}\\
		&=\int_{[0,1]^p}\left(f_n(\v{x})-\sum\limits_{j=1}^dx_j\beta_{Z_p}(\s{A}_0)_j-\widetilde{\beta}_0\right)^2d\v{x}-\int_{[0,1]^p}\left(f_n(\v{x})-\sum\limits_{j=1}^{p_0}x_j\beta_{Z_p}(\s{A}_0)_j
		-\beta_0({\s{A}_0})\right)^2d\v{x}\\
		&\quad\quad+\int_{[0,1]^p}\left(\sum\limits_{j=p_0+1}^{p_0+M-d}x_j\beta_{Z_p}(\s{A}_1)_j-\frac{1}{2}\sum\limits_{j=p_0+1}^{p_0+M-d}\beta_{Z_p}(\s{A}_1)_j\right)^2d\v{x}\\
		&\quad\quad-2\int_{[0,1]^p}\left(f_n(\v{x})-\sum\limits_{j=1}^dx_j\beta_{Z_p}(\s{A}_0)_j-\widetilde{\beta}_0\right)\left(\sum\limits_{j=p_0+1}^{p_0+M-d}x_j\beta_{Z_p}(\s{A}_1)_j
		-\frac{1}{2}\sum\limits_{j=p_0+1}^{p_0+M-d}\beta_{Z_p}(\s{A}_1)_j\right)d\v{x}\\
		&=\int_{[0,1]^p}\sum\limits_{j=d+1}^{p_0}(2x_j-1)\beta_{Z_p}(\s{A}_0)_jf_n(\v{x})d\v{x}
		-\frac{1}{4}\int_{[0,1]^p}\left(\sum\limits_{j=d+1}^{p_0}(2x_j-1)\beta_{Z_p}(\s{A}_0)_j\right)^2d\v{x}\notag\\
		&\quad\quad+\frac{1}{4}\int_{[0,1]^p}\left(\sum\limits_{j=p_0+1}^{p_0+M-d}\left(2x_j-1\right)\beta_{Z_p}(\s{A}_1)_j\right)^2d\v{x}
		-\int_{[0,1]^p}\sum\limits_{j=p_0+1}^{p_0+M-d}(2x_j-1)\beta_{Z_p}(\s{A}_1)_jf_n(\v{x})d\v{x}\notag\\
		&=\frac{1}{12}\sum\limits_{j=d+1}^{p_0}{\beta_{Z_p}(\s{A}_0)_j}^2-\frac{1}{12}\sum\limits_{j=p_0+1}^{p_0+M-d}{\beta_{Z_p}(\s{A}_1)_j}^2\\
		&\geqslant 12(p_0-d)\tau^2-12(M-d)\eta_n^2\ \ (\text{by \eqref{eb1} and \eqref{eb2}})\\
		&\geqslant 12\tau^2-12(M-p_0+1)\eta_n^2,
		\end{align*}where $\widetilde{\beta}_0=\int_{[0,1]^p}f_n(\v{x})d\v{x}-\sum\limits_{j=1}^d\beta_{Z_p}(\s{A}_1)_j/2$.
		This completes the proof.
	\end{proof}

	\begin{lemma}\label{lemma:cof}
		Under Assumptions 1, 2, and 3, we have
		\begin{align*}
		\left\| \begin{pmatrix}
		\beta_0({\s{A}_0})\\\vg{\beta}({\s{A}_0})
		\end{pmatrix}-\begin{pmatrix}
		\widehat{\beta}_0({\s{A}_0})\\\widehat{\vg{\beta}}({\s{A}_0})
		\end{pmatrix}\right\|
		&\leqslant C_{1n}(p_0+1)(9p_0+1)\Delta_{p_0,n}+2(p_0+1)(9p_0+1)\eta_n\\
		&+\frac{3}{1-\alpha}C_{2n}(p_0+1)^3(9p_0+1)^2\Delta_{2,n}
		\end{align*}
		for sufficiently large $n$, where $(\beta_0({\s{A}_0}),\vg{\beta}({\s{A}_0})')'$ and $(\widehat{\beta}_0({\s{A}_0}),\widehat{\vg{\beta}}({\s{A}_0})')'$  are defined by (7) and (9), respectively.
	\end{lemma}
	\begin{proof}	
		By (\ref{betaZw}), we have
		\begin{align*}
		\begin{pmatrix}
		\beta_0({\s{A}_0})\\\vg{\beta}({\s{A}_0})
		\end{pmatrix}=\m{Z}^{-1}\v{w},
		\end{align*}
		where $\m{Z}$ and $\v{w}$ are defined as $\m{U}$ and $\v{v}$ in (\ref{Zw}), respectively.
		Note that for any matrix $\m{A}=(a_{ij}) \in \mathbb{R}^{m \times n}$,
		\begin{align} \label{inequalitymatrix}
		\|\m{A}\| \leqslant \sqrt{m}\|\m{A}\|_\infty,
		\end{align}
		where $\|\m{A}\|_\infty=\max\{\sum_{j=1}^n |a_{1j}|,\ldots,\sum_{j=1}^n |a_{mj}|\}$ (Golub and Van Loan 1996).
		By (\ref{Zn}) and (\ref{inequalitymatrix}), we get
		\begin{align}\label{A.14}
		\|\m{Z}^{-1}\| \leqslant (p_0+1)^\frac{1}{2}(9p_0+1).
		\end{align}
		By (9), we have
		\begin{align*}
		\begin{pmatrix}
		\widehat{\beta}_0({\s{A}_0})\\\widehat{\vg{\beta}}({\s{A}_0})
		\end{pmatrix}=\widehat{\m{Z}}^{-1}\widehat{\v{w}},
		\end{align*}
		where
		\begin{align}\label{hatzw}
		\widehat{\m{Z}}=
		\begin{pmatrix}
		1 && \frac{\sum\limits_{i=1}^nx_{i1}}{n}  && \cdots &&  \frac{\sum\limits_{i=1}^nx_{ip_0}}{n}\\
		\frac{\sum\limits_{i=1}^nx_{i1}}{n} && \frac{\sum\limits_{i=1}^nx_{i1}^2}{n} &&  \cdots &&\frac{\sum\limits_{i=1}^nx_{i1} x_{ip_0}}{n}\\
		\vdots && \vdots && \ddots && \vdots\\
		\frac{\sum\limits_{i=1}^nx_{ip_0}}{n} && \frac{\sum\limits_{i=1}^nx_{i1}x_{ip_0}}{n} &&  \cdots &&\frac{\sum\limits_{i=1}^nx_{ip_0}^2}{n}
		\end{pmatrix}_{(p_0+1)\times (p_0+1)},
		\widehat{\v{w}}=
		\begin{pmatrix}
		\frac{\sum\limits_{i=1}^n y_i}{n}\\
		\frac{\sum\limits_{i=1}^n x_{i1}y_i}{n}\\
		\vdots\\
		\frac{\sum\limits_{i=1}^n x_{ip_0}y_i}{n}
		\end{pmatrix}.
		\end{align}

		It follows from Lemma \ref{lemma2} that $$\left|\int_{[0,1]^{p_0}}\widetilde{f}_n(\v{x}_{\s{A}_0})d\v{x}_{\s{A}_0}-\frac{\sum\limits_{i=1}^n\widetilde{f}_n(x_{i1},\ldots,x_{ip_0})}{n}\right|\leqslant \Delta_{p_0,n}V_{HK}(\widetilde{f}_n),$$where $\Delta_{p_0,n}$ is defined in Section 3. By (5), $\sup_{\v{x} \in [0,1]^p} |r_n(\v{x})|\leqslant \eta_n$, where $r_n$ is defined in \eqref{r}. We have 	 $$\int_{[0,1]^p}f_n(\v{x})d\v{x}=\int_{[0,1]^{p_0}}\widetilde{f}_n(\v{x}_{\s{A}_0})d\v{x}_{\s{A}_0}+\int_{[0,1]^p}r_n(\v{x})d\v{x},$$and	 $$\frac{\sum\limits_{i=1}^ny_i}{n}=\frac{\sum\limits_{i=1}^nf_n(\v{x}_i)}{n}=\frac{\sum\limits_{i=1}^n\widetilde{f}_n(x_{i1},\ldots,x_{ip_0})}{n}
		+\frac{\sum\limits_{i=1}^nr_n(\v{x}_i)}{n}.$$
		Therefore. \begin{equation}\label{vb}\left|\int_{[0,1]^p}f_n(\v{x})d\v{x}- \frac{\sum\limits_{i=1}^ny_i}{n}\right|<\Delta_{p_0,n}V_{HK}(\widetilde{f}_n)+2\eta_n.\end{equation}
		On the other hand, by Lemma \ref{lemma3}, Lemma \ref{lemma:owen2}, and Assumption 2, $x_j\widetilde{f}_n \in$ BVHK for $j=1,\ldots,p$. Similar to \eqref{vb}, \begin{equation}\label{vbh}
		\left|\int_{[0,1]^p}x_jf_n(\v{x})d\v{x}- \frac{\sum\limits_{i=1}^nx_{ij}y_i}{n}\right|<\Delta_{p_0,n}V_{HK}(x_j\widetilde{f}_n)+2\eta_n,\ j=1,\ldots,p_0.\end{equation}Combining \eqref{vb} and \eqref{vbh}, by \eqref{inequalitymatrix},  we have
		\begin{align}\label{A.15}
		\lVert\v{w}-\widehat{\v{w}}\rVert < (p_0+1)^\frac{1}{2}\left(C_{1n}\Delta_{p_0,n}+2\eta_n\right),
		\end{align}where $C_{1n}$ is defined in Section 3.
		
		By Lemma 3, for $h(x_1,x_2)=x_1x_2,\ x_1,x_2\in[0,1]$,
		\begin{align}\label{A.16}
		V_{HK}(h)\leqslant 3.
		\end{align}
		By (\ref{inequalitymatrix}), (\ref{A.16})  and
		Lemma \ref{lemma2}, we have
		\begin{align}\label{A.17}
		\lVert\m{Z}-\widehat{\m{Z}}\rVert \leqslant 3(p_0+1)^\frac{3}{2} \Delta_{2,n}.
		\end{align}
		By Lemma \ref{lemma1}, we get
		\begin{align}\label{A.18}
		|\lambda_{\min}(\m{Z})-\lambda_{\min}(\widehat{\m{Z}})|\leqslant \|\m{Z}-\widehat{\m{Z}}\|,
		\end{align}
		where $\lambda_{\min}(\cdot)$ represents the minimum eigenvalue of a matrix.
		
		By (\ref{A.14}), (\ref{A.17}), (\ref{A.18}), and Assumption 3, we have
		\begin{align}
		\left\lVert\m{Z}^{-1}-\widehat{\m{Z}}^{-1}\right\rVert &\leqslant \left\lVert\m{Z}^{-1}\right\rVert \left\lVert\m{Z}-\widehat{\m{Z}}\right\rVert \left\lVert\widehat{\m{Z}}^{-1}\right\rVert \notag
		=\left\|\m{Z}^{-1}\right\|\left\|\m{Z}-\widehat{\m{Z}}\right\| \frac{1}{\lambda_{\min}(\widehat{\m{Z}})}\notag\\
		&\leqslant \left\|\m{Z}^{-1}\right\|\left\|\m{Z}-\widehat{\m{Z}}\right\|\frac{1}{\lambda_{\min}(\m{Z})-\left\|\m{Z}-\widehat{\m{Z}}\right\|}\notag=\frac{\left\|\m{Z}^{-1}\right\|^2\left\|\m{Z}
			-\widehat{\m{Z}}\right\|}{1-\left\|\m{Z}^{-1}\right\|\left\|\m{Z}-\widehat{\m{Z}}\right\|}\notag\\
		&\leqslant \frac{3}{1-\alpha}(p_0+1)^\frac{5}{2}(9p_0+1)^2\Delta_{2,n}\label{A.19}
		\end{align}
		for sufficiently large $n$.
		
		Note that $\left|\sum\limits_{i=1}^n y_i/n\right|<C_{2n}$, and that $\left|\sum\limits_{i=1}^n x_{ij}y_i/n\right|<C_{2n}$ since $x_{ij} \in [0,1),\ i=1,\cdots,n,\ j=1,\cdots,p$, where $C_{2n}$ is defined in Section 3.
		By (\ref{inequalitymatrix}) and \eqref{hatzw}, we have
		\begin{align} \label{w}
		\left	\|\widehat{\v{w}}\right\|\leqslant (p_0+1)^\frac{1}{2}C_{2n} .
		\end{align}
		Thus, by (\ref{A.14}), (\ref{A.15}), (\ref{A.18}), and (\ref{w}), we have
		\begin{align*}
		&\left\lVert\begin{pmatrix}
		\beta_0({\s{A}_0})\\\vg{\beta}({\s{A}_0})
		\end{pmatrix}-\begin{pmatrix}
		\widehat{\beta}_0({\s{A}_0})\\\widehat{\vg{\beta}}({\s{A}_0})
		\end{pmatrix}\right\Vert=\left\lVert\m{Z}^{-1}\v{w}-\v{\widehat{Z}}^{-1}\v{\widehat{w}}\right\rVert
		\\&=\left\lVert\m{Z}^{-1}\v{w}-\m{Z}^{-1}\v{\widehat{w}}+\m{Z}^{-1}\v{\widehat{w}}-\v{\widehat{Z}}^{-1}\v{\widehat{w}}\right\rVert
		\leqslant \left\lVert\m{Z}^{-1}\right\rVert \left\lVert\v{w}-\widehat{\v{w}}\right\rVert+\left\lVert\m{Z}^{-1}-\m{\widehat{Z}}^{-1}\right\rVert \left\lVert\v{\widehat{w}}\right\rVert\\
		&\leqslant (p_0+1)(9p_0+1)(C_{1n}\Delta_{p_0,n}+2\eta_n)+\frac{3}{1-\alpha}C_{2n}(p_0+1)^3(9p_0+1)^2\Delta_{2,n} \\
		&=C_{1n}(p_0+1)(9p_0+1)\Delta_{p_0,n}+2(p_0+1)(9p_0+1)\eta_n+\frac{3}{1-\alpha}C_{2n}(p_0+1)^3(9p_0+1)^2\Delta_{2,n}
		\end{align*}
		for sufficiently large $n$.
	\end{proof}

	\begin{lemma}\label{lemma del}
		Under Assumptions 5 and 6, for $\gamma_0$ with $0<\gamma_0<\tilde{\gamma}_0$, as $n\to\infty$,
		\begin{align*}
		P\left(\{\Delta_{p_0,n} \leqslant n^{-\gamma_0}\} \cap \{\Delta_{2,n} \leqslant n^{-\gamma_0}\} \cap \{\Delta_{p_0+1,n} \leqslant n^{-\gamma_0}\} \cap \{\Delta_{p_0+M-1,n} \leqslant n^{-\gamma_0}\} \right) \to 1.
		\end{align*}
	\end{lemma}
	\begin{proof}
		Note that in Lemma \ref{lemma4}, $\sup\limits_{\v{x}}\left |F(\v{x})-F_n(\v{x})\right|$ becomes the $L_\infty$ discrepancy $\nu_{p,n}$ when $F$ is the uniform distribution on $[0,1]^p$.
		
		By Lemma \ref{lemma4}, since $1-2\gamma_0>0$ in Assumption 6, we have
		\begin{align}
		P(\Delta_{p_0,n} \leqslant n^{-\gamma_0})=1-O(\exp{(-D_1n^{(1-2\gamma_0)})})\to 1,\label{A.27}
		\end{align}
		where $D_1$ is a positive constant.
		
		Furthermore, by Lemma \ref{lemma4}, since $\gamma_{10}+\ 2\gamma_0<1$ in Assumption 6, we have
		\begin{align}
		P(\Delta_{2,n} \leqslant n^{-\gamma_0}) &\geqslant 1-\sum_{\s{A} \in \s{U}_2}P(\delta_{2,n}(\s{A}) \leqslant n^{-\gamma_0})\notag \geqslant 1-\binom{p}{2}D_0\exp\left(-D_2n^{(1-2\gamma_0)}\right) \notag\\
		&=1-p^2O\left(\exp{\left(-D_2n^{(1-2\gamma_0)}\right)}\right)=1-O\left(\exp{\left(2n^{\gamma_{10}}-D_2n^{(1-2\gamma_0)}\right)}\right)\notag\\&\to 1,\label{A.28}
		\end{align}
		where $D_0$ and $D_2$ are positive constants.
		
		Similarly, we have
		\begin{align}
		P(\Delta_{p_0+1,n} \leqslant n^{-\gamma_0}) &\geqslant 1-\sum_{\s{A} \in \s{U}_3}P\left(\delta_{p_0+1,n}(\s{A}) \leqslant n^{-\gamma_0}\right)\notag\\ &=1-(p-p_0)O\left(\exp{\left(-D_3n^{(1-2\gamma_0)}\right)}\right)\to 1,\label{A.29}
		\end{align}and
		\begin{align}
		P(\Delta_{p_0+M-1,n} \leqslant n^{-\gamma_0}) &\geqslant 1-\sum_{\s{A} \in \s{U}_4}P(\delta_{p_0+M-1,n}(\s{A}) \leqslant n^{-\gamma_0}) \notag
		\\&=1-(p-p_0)^{(M-1)}O\left(\exp{\left(-D_4n^{(1-2\gamma_0)}\right)}\right)\to 1,\label{A.30}
		\end{align}
		where $D_3$ and $D_4$ are positive constants.
		
		Combining (\ref{A.27}), (\ref{A.28}), (\ref{A.29}), and (\ref{A.30}), we complete the proof.
	\end{proof}

	\subsection{Proofs of theorems}
	\begin{proof}{(Proof of Theorem 1)}		
		By (3), (4), and Assumption 1, this proof is similar to those of \eqref{eb1} and \eqref{eb2}.
	\end{proof}

	\begin{proof}{(Proof of Theorem 2)}\\
		Note that for matrices $\m{A},\m{B} \in \mathbb{R}^{n\times n}$ (Golub and Van Loan 1996),
		\begin{align}\label{matrixtimes}
		\|\m{A}\m{B}\| \leqslant \|\m{A}\|\|\m{B}\|.
		\end{align}
		Since all the elements of $\m{X}$ lie in $[0,1)$, by (\ref{inequalitymatrix}), we get
		\begin{align} \label{B.1}
		\left \| \begin{pmatrix}
		\vg{1}_n  & \m{X}
		\end{pmatrix}
		\right	\|  \leqslant \sqrt{n}(p_0+1),
		\end{align}
		Therefore. by
		(\ref{matrixtimes}), (\ref{B.1}), and Lemma \ref{lemma:cof}, we have
		\begin{align}
		&\frac{1}{n}\left\|\vg{1}_n\beta_0({\s{A}_0})+\m{X}\vg{\beta}_{Z_p}({\s{A}_0})-\vg{1}_n\widehat{\beta}_0({\s{A}_0})-\m{X}\widehat{\vg{\beta}}_{Z_p}({\s{A}_0})\right\|^2 \notag\\ & \leqslant \frac{1}{n}\left \| \begin{pmatrix}
		\vg{1}_n  & \m{X}
		\end{pmatrix}
		\right	\|^2 \left \| \begin{pmatrix}
		\beta_0({\s{A}_0})\\\vg{\beta}_{Z_p}({\s{A}_0})
		\end{pmatrix}-\begin{pmatrix}
		\widehat{\beta}_0({\s{A}_0})\\\widehat{\vg{\beta}}_{Z_p}({\s{A}_0})
		\end{pmatrix}\right	\|^2 \leqslant \zeta_{1n}^2\label{A.20}
		\end{align}
		for sufficiently large $n$,
		where $\zeta_{1n}$ is defined in Section 3.
		By Assumption 1 and the definition of $C_{3n}$ in Section 3,
		\begin{align}\label{A.21}
		\max\limits_ {i=1,\ldots,n}\left|y_i-\beta_0({\s{A}_0})-\vg{\beta}_{Z_p}({\s{A}_0})'\v{x}_i\right|<C_{3n}+\eta_n.
		\end{align}
		By \eqref{inequalitymatrix}, (\ref{A.20}), and (\ref{A.21}),
		\begin{align}
		&\frac{2}{n}\left\|\vg{1}_n\beta_0({\s{A}_0})+\m{X}\vg{\beta}_{Z_p}({\s{A}_0})-\vg{1}_n\widehat{\beta}_0({\s{A}_0})-\m{X}\widehat{\vg{\beta}}_{Z_p}({\s{A}_0})\right\| \left\|\v{y}-\vg{1}_n\beta_0({\s{A}_0})-\m{X}\vg{\beta}_{Z_p}({\s{A}_0})\right\|\notag\\
		&\leqslant 2(C_{3n}+\eta_n)\zeta_{1n}.\label{A.22}
		\end{align}
		Combining (\ref{A.20}), (\ref{A.21}), and (\ref{A.22}), we have
		\begin{align}
		&\frac{1}{n}\left\|\v{y}-\vg{1}_n\widehat{\beta}_0({\s{A}_0})-\m{X}\widehat{\vg{\beta}}_{Z_p}({\s{A}_0})\right\|^2\notag\\
		&=\frac{1}{n}\left\|\v{y}-\vg{1}_n\beta_0({\s{A}_0})-\m{X}\vg{\beta}_{Z_p+}({\s{A}_0})+\vg{1}_n\beta_0({\s{A}_0})
		+\m{X}\vg{\beta}_{Z_p}({\s{A}_0})-\vg{1}_n\widehat{\beta}_0({\s{A}_0})-\m{X}\widehat{\vg{\beta}}_{Z_p}({\s{A}_0})\right\|^2\notag\\
		&\leqslant\frac{1}{n}\left\|\v{y}-\vg{1}_n\beta_0({\s{A}_0})-\m{X}\vg{\beta}_{Z_p}({\s{A}_0})\right\|^2+\frac{1}{n}\left\|\vg{1}_n\beta_0({\s{A}_0})
		+\m{X}\vg{\beta}_{Z_p}({\s{A}_0})-\vg{1}_n\widehat{\beta}_0({\s{A}_0})-\m{X}\widehat{\vg{\beta}}_{Z_p}({\s{A}_0})\right\|^2\notag\\
		&+\frac{2}{n}\left\|\vg{1}_n\beta_0({\s{A}_0})+\m{X}\vg{\beta}_{Z_p}({\s{A}_0})-\vg{1}_n\widehat{\beta}_0({\s{A}_0})
		-\m{X}\widehat{\vg{\beta}}_{Z_p}({\s{A}_0})\right\|\left\|\v{y}-\vg{1}_n\beta_0({\s{A}_0})-\m{X}\vg{\beta}_{Z_p}({\s{A}_0})\right\|\notag
		\\
		&\leqslant\frac{1}{n}\left\|\v{y}-\vg{1}_n\beta_0({\s{A}_0})-\m{X}\vg{\beta}_{Z_p}({\s{A}_0})\right\|^2+\zeta_{1n}^2 + 2\left(C_{3n}+\eta_n\right)\zeta_{1n}\notag\\
		&\leqslant \frac{1}{n}\left\|\v{y}-\vg{1}_n\beta_0({\s{A}_0})-\m{X}\vg{\beta}_{Z_p}({\s{A}_0})\right\|^2+\rho_{1n}\label{A.23}
		\end{align}
		for sufficiently large $n$.
		In addition, we have
		\begin{align}
		&\max\limits_ {i=1,\ldots,n}\left|\beta_0(\s{A}_0)+\vg{\beta}_{Z_p}({\s{A}_0})'\v{x}_i-\widehat{\beta}_0(\s{A}_0)-\widehat{\vg{\beta}}_{Z_p}({\s{A}_0})'\v{x}_i\right| \notag \\
		& \leqslant \left|\beta_0(\s{A}_0)-\widehat{\beta}_0(\s{A}_0)\right| + \sum_{j=1}^{p_0}\left|\beta_{Z_p}({\s{A}_0})_j-\widehat{\beta}_{Z_p}({\s{A}_0})_j\right| \notag\\
		& \leqslant \left(p_0+1\right)^{1/2}\left\| \begin{pmatrix}
		\beta_0({\s{A}_0})\\\vg{\beta}({\s{A}_0})
		\end{pmatrix}-\begin{pmatrix}
		\widehat{\beta}_0({\s{A}_0}) \notag\\
		\widehat{\vg{\beta}}({\s{A}_0})
		\end{pmatrix}\right\| \notag\\
		& \leqslant (p_0+1)^{-1/2}\zeta_{1n}\quad(\text{By Lemma \ref{lemma:cof}}) \label{A.30}
		\end{align}
		for sufficiently large $n$.
		By (\ref{A.21}) and (\ref{A.30}), for sufficiently large $n$,
		\begin{align*}
		&\max\limits_ {i=1,\ldots,n}\left|y_i-\widehat{\beta}_0(\s{A}_0)-\widehat{\vg{\beta}}_{Z_p}({\s{A}_0})'\v{x}_i\right|\\
		&\leqslant \max\limits_ {i=1,\ldots,n}\left|y_i-\beta_0(\s{A}_0)-\vg{\beta}_{Z_p}({\s{A}_0})'\v{x}_i\right|+	 \max\limits_ {i=1,\ldots,n}\left|\beta_0(\s{A}_0)+\vg{\beta}_{Z_p}({\s{A}_0})'\v{x}_i-\widehat{\beta}_0(\s{A}_0)-\widehat{\vg{\beta}}_{Z_p}({\s{A}_0})'\v{x}_i\right|\\
		&<C_{3n}+\eta_n+(p_0+1)^{-1/2}\zeta_{1n}.
		\end{align*} Similar to (\ref{A.23}), we have
		\begin{align}
		&\frac{1}{n}\left\|\v{y}-\vg{1}_n\beta_0({\s{A}_0})-\m{X}\vg{\beta}_{Z_p}({\s{A}_0})\right\|^2\notag\\
		&=\frac{1}{n}\left\|\v{y}-\vg{1}_n\widehat{\beta}_0({\s{A}_0})-\m{X}\widehat{\vg{\beta}}_{Z_p}({\s{A}_0})+\vg{1}_n\widehat{\beta}_0({\s{A}_0})
		-\m{X}\widehat{\vg{\beta}}_{Z_p}({\s{A}_0})
		-\vg{1}_n\beta_0({\s{A}_0})-\m{X}\vg{\beta}_{Z_p}({\s{A}_0})\right\|^2\notag\\
		&\leqslant\frac{1}{n}\left\|\v{y}-\vg{1}_n\widehat{\beta}_0({\s{A}_0})-\m{X}\widehat{\vg{\beta}}_{Z_p}({\s{A}_0})\right\|^2
		+\frac{1}{n}\left\|\vg{1}_n\widehat{\beta}_0({\s{A}_0})-\m{X}\widehat{\vg{\beta}}_{Z_p}({\s{A}_0})
		-\vg{1}_n\beta_0({\s{A}_0})-\m{X}\vg{\beta}_{Z_p}({\s{A}_0})\right\|^2\notag\\
		&+\frac{2}{n}\left\|\v{y}-\vg{1}_n\widehat{\beta}_0({\s{A}_0})-\m{X}\widehat{\vg{\beta}}_{Z_p}({\s{A}_0})\right\|\left\|\vg{1}_n\widehat{\beta}_0({\s{A}_0})
		-\m{X}\widehat{\vg{\beta}}_{Z_p}({\s{A}_0})
		-\vg{1}_n\beta_0({\s{A}_0})-\m{X}\vg{\beta}_{Z_p}({\s{A}_0})\right\| \notag\\
		&\leqslant \frac{1}{n}\left\|\v{y}-\vg{1}_n\widehat{\beta}_0({\s{A}_0})-\m{X}\widehat{\vg{\beta}}_{Z_p}({\s{A}_0})\right\|^2+\zeta_{1n}^2
		+2(C_{3n}+\eta_n)\zeta_{1n}+2(p_0+1)^{-1/2}\zeta_{1n}^2\notag\\
		&=\frac{1}{n}\left\|\v{y}-\vg{1}_n\widehat{\beta}_0({\s{A}_0})-\m{X}\widehat{\vg{\beta}}_{Z_p}({\s{A}_0})\right\|^2+\rho_{1n}.
		\label{A.24}
		\end{align}
		Combining (\ref{A.23}) and (\ref{A.24}), we get
		\begin{align}\label{A.25}
		\left| \frac{1}{n}\left\|\v{y}-\vg{1}_n\beta_0(\s{A}_0)-\m{X}\vg{\beta}_{Z_p}({\s{A}_0})\right\|^2-\frac{1}{n}\left\|\v{y}-\vg{1}_n\widehat{\beta}_0(\s{A}_0)
		-\m{X}\widehat{\vg{\beta}}_{Z_p}({\s{A}_0})\right\|^2\right|\leqslant \rho_{1n}
		\end{align}
		for sufficiently large $n$.

		Besides, by Assumption 1 and the definition of $C_{3n}$ in Section 3, we have
		\begin{align}
		&\left|\frac{1}{n}\left\|\v{y}-\vg{1}_n\beta_0(\s{A}_0)-\m{X}\vg{\beta}_{Z_p}({\s{A}_0})\right\|^2-\frac{1}{n}\left\|\widetilde{\vg{f}}_n(\m{X}_{\s{A}_0})
		-\vg{1}_n\beta_0(\s{A}_0)-\m{X}\vg{\beta}_{Z_p}({\s{A}_0})\right\|^2\right| \notag\\
		&\leqslant \frac{1}{n}\left\|\v{y}-\widetilde{\vg{f}}_n(\m{X}_{\s{A}_0})\right\|^2 +\frac{2}{n}\left\|\v{y}-\widetilde{\vg{f}}_n(\m{X}_{\s{A}_0})\right\|\left\|\widetilde{\vg{f}}_n(\m{X}_{\s{A}_0})-\vg{1}_n\beta_0(\s{A}_0)-\m{X}\vg{\beta}_{Z_p}({\s{A}_0})\right\|\notag\\
		&\leqslant\eta_n^2+2\eta_n C_{3n}, \label{A.33}
		\end{align}where $\widetilde{\vg{f}}_n(\v{X}_{\s{A}_0})=\begin{pmatrix}
		\widetilde{f}_n({x_{11},\cdots,x_{1p_0}}),\cdots,
		\widetilde{f}_n({x_{n1},\cdots,x_{np_0}})
		\end{pmatrix}'$,
		and
		\begin{align}
		&
		\left |\int_{[0,1]^p}\left(f_n(\v{x})-\beta_0(\s{A}_0)-\v{x}'\vg{\beta}_{Z_p}(\s{A}_0)\right)^2d\v{x}-\int_{[0,1]^{p_0}}\left(\widetilde{f}_n(\v{x}_{\s{A}_0})-\beta_0(\s{A}_0)
		-\v{x}'\vg{\beta}_{Z_p}(\s{A}_0)\right)^2d\v{x}_{\s{A}_0}\right| \notag\\
		&\leqslant\int_{[0,1]^p}\left(f_n(\v{x})-\widetilde{f}_n(\v{x}_{\s{A}_0})\right)^2d\v{x}+2\int_{[0,1]^p}\left|f_n(\v{x})
		-\widetilde{f}_n(\v{x}_{\s{A}_0})\right|\left|\widetilde{f}_n(\v{x}_{\s{A}_0})
		-\beta_0(\s{A}_0)-\v{x}'\vg{\beta}_{Z_p}(\s{A}_0)\right|d\v{x} \notag\\
		&\leqslant\eta_n^2+2\eta_n C_{3n}. \label{A.34}
		\end{align}
		By Assumption 2, Lemma \ref{lemma3}, and Lemma \ref{lemma:owen2}, $\widetilde{f}_n(\v{x}_{\s{A}_0})-\beta_0(\s{A}_0)-\v{x}'\vg{\beta}_{Z_p}(\s{A}_0) \in$ BVHK, Thus, by Lemma \ref{lemma2}. we have
		\begin{align} \label{A.35}
		&\left|\frac{1}{n}\left\|\widetilde{\vg{f}}_n(\v{X}_{\s{A}_0})-\vg{1}_n\beta_0(\s{A}_0)-\m{X}\vg{\beta}_{Z_p}(\s{A}_0)\right\|^2
		-\int_{[0,1]^{p_0}}\left(\widetilde{f}_n(\v{x}_{\s{A}_0})-\beta_0(\s{A}_0)
		-\v{x}'\vg{\beta}_{Z_p}(\s{A}_0)\right)d\v{x}\right| \notag
		\\&\leqslant \Delta_{p_0,n}V_{1n}.
		\end{align}
		By (\ref{A.33}), (\ref{A.34}), and (\ref{A.35}),
		\begin{align}
		&\left|\frac{1}{n}\left\|\v{y}-\vg{1}_n\beta_0(\s{A}_0)-\m{X}\vg{\beta}_{Z_p}({\s{A}_0})\right\|^2-\int_{[0,1]^p}\left(f_n(\v{x})-\beta_0(\s{A}_0)
		-\v{x}'\vg{\beta}_{Z_p}(\s{A}_0)\right)^2d\v{x}\right| \notag\\
		&\leqslant 2(\eta_n^2+2\eta_n C_{3n})+\Delta_{p_0,n}V_{1n},\label{A.26}
		\end{align}
		Combing (\ref{A.25}) and (\ref{A.26}), we get
		\begin{align}
		&\left|\frac{1}{n}\left\|\v{y}-\vg{1}_n\widehat{\beta}_0(\s{A}_0)-\m{X}\widehat{\vg{\beta}}_{Z_p}({\s{A}_0})\right\|^2-\int_{[0,1]^p}\left(f_n(\v{x})
		-\beta_0(\s{A}_0)-\v{x}'\vg{\beta}_{Z_p}(\s{A}_0)\right)^2d\v{x}\right| \notag\\
		&\leqslant2\left(\eta_n^2+2\eta_n C_{3n}\right)+\Delta_{p_0,n}V_{1n}+\rho_{1n}, \label{f1}
		\end{align}
		
		Now consider any $\s{A}_1 \in \s{U}_1$. Similar to the proof of Lemma \ref{lemma:cof}, we have
		\begin{align*}
		\left\|\begin{pmatrix}
		\widehat{\beta}_0(\s{A}_1)\\\widehat{\vg{\beta}}({\s{A}_1})
		\end{pmatrix}-
		\begin{pmatrix}
		\beta_0(\s{A}_1)\\\vg{\beta}({\s{A}_1})
		\end{pmatrix}\right\| \leqslant& C_{1n}(M+1)(9M+1)\Delta_{p_0+1,n}+2(M+1)(9M+1)\eta_n\\+&\frac{3}{1-\alpha}C_{2n}(M+1)^3(9M+1)^2\Delta_{2,n}
		\end{align*}
		for sufficiently large $n$. In addition, similar to the proof of (\ref{f1}), we have
		\begin{align}
		&
		\left|\frac{1}{n}\left\|\v{y}-\vg{1}_n\widehat{\beta}_0(\s{A}_1)-\v{X}\widehat{\vg{\beta}}_{Z_p}({\s{A}_1})\right\|^2-\int_{[0,1]^p}\left(f_n(\v{x})
		-\beta_0(\s{A}_1)-\v{x}'\vg{\beta}_{Z_p}(\s{A}_1)\right)^2d\v{x}\right|  \notag\\
		&\leqslant2(\eta_n^2+2\eta_nC_{4n})+\Delta_{p_0+M-1,n}V_{2n}+\rho_{2n} \label{f2}
		\end{align}
		for sufficiently large $n$.\\
		By (\ref{f1}), (\ref{f2}), Lemma \ref{lemma:a0}, and Assumption 4,
		\begin{align*}
		&	\frac{1}{n}\left\|\v{y}-\vg{1}_n\widehat{\beta}_0(\s{A}_1)-\m{X}\widehat{\vg{\beta}}_{Z_p}({\s{A}_1})\right\|^2
		-\frac{1}{n}\left\|\v{y}-\vg{1}_n\widehat{\beta}_0(\s{A}_0)-\v{X}\widehat{\vg{\beta}}_{Z_p}({\s{A}_0})\right\|^2\\
		&\geqslant\int_{[0,1]^p}\left(f_n(\v{x})-\beta_0(\s{A}_1)-\v{x}'\vg{\beta}_{Z_p}(\s{A}_1)\right)^2d\v{x}-2\left(\eta_n^2+2\eta_n C_{4n}\right)-\Delta_{p_0+M-1,n}V_{2n}-\rho_{2n}\\&
		-\int_{[0,1]^p}\left(f_n(\v{x})-\beta_0(\s{A}_0)-\v{x}'\vg{\beta}_{Z_p}(\s{A}_0)\right)^2d\v{x}-2\left(\eta_n^2+2\eta_n C_{3n}\right)-\Delta_{p_0,n}V_{1n}-\rho_{1n}\\
		&\geqslant 12\tau^2-12\left(M-p_0+1\right)\eta_n^2-4\eta_n^2-4\eta_n\left(C_{4n}+C_{3n}\right)-\Delta_{p_0,n}V_{1n}-\Delta_{p_0+M-1,n}V_{2n}-\rho_{1n}-\rho_{2n}\\&>D,
		\end{align*}
		which implies
		\begin{align}\label{meq}
		\frac{1}{n}\left\|\v{y}-\vg{1}_n\widehat{\beta}_0(\s{A}_0)-\m{X}\widehat{\vg{\beta}}_{Z_p}({\s{A}_0})\right\|^2
		<\frac{1}{n}\left\|\v{y}-\vg{1}_n\widehat{\beta}_0(\s{A}_1)-\m{X}\widehat{\vg{\beta}}_{Z_p}({\s{A}_1})\right\|^2
		\end{align}for sufficiently large $n$.
		
		Next consider any $\s{A}\in \s{U}_0$. Since $\s{A}\supset\s{A}_0$, we have
		\begin{align}\label{oeq}
		\frac{1}{n}\left\|\v{y}-\vg{1}_n\widehat{\beta}_0(\s{A})-\m{X}\widehat{\vg{\beta}}_{Z_p}(\s{A})\right\|^2\leqslant\frac{1}{n}\left\|\v{y}-\vg{1}_n\widehat{\beta}_0(\s{A}_0)
		-\m{X}\widehat{\vg{\beta}}_{Z_p}({\s{A}_0})\right\|^2.
		\end{align}
		Combining \eqref{meq} and \eqref{oeq}, we complete the proof.
	\end{proof}
	
	\begin{proof} (Proof of Theorem 3)
		Since $3\gamma_9<\gamma_0$, we have $3(M+1)^2(9M+1)\Delta_{2,n}=O(n^{3\gamma_9-\gamma_0}) \to 0$ as $n\to\infty$, and this implies Assumption 3.
		
		Let $\alpha=1/2$ in Assumption 3. It follows from $\gamma_9<2\gamma_1$ that $12(M-p_0+1)\eta_n^2 =O(n^{\gamma_9-2\gamma_1})\to 0$. By $\gamma_7<\gamma_1, \ \gamma_6<\gamma_1$, we have $4\eta_n(C_{4n}+C_{3n})=O(n^{\gamma_7-\gamma_1})+O(n^{\gamma_6-\gamma_1}) \to 0$.
		By $\gamma_2<\gamma_0,\gamma_3<\gamma_0$, we have $\Delta_{p_0,n}V_{1n}+\Delta_{p_0+M-1,n}V_{2n} =O(n^{\gamma_2-\gamma_0})+O(n^{\gamma_3-\gamma_0})\to 0$. Similarly, by $3\gamma_9+\gamma_4<\gamma_0,\ 3\gamma_9<\gamma_1,\ 6\gamma_9+\gamma_5<\gamma_0,\ 3\gamma_8+\gamma_4+\gamma_6<\gamma_0,\ 3\gamma_8+\gamma_6<\gamma_1,\ 6\gamma_8+\gamma_5+\gamma_6<\gamma_0,\
		3\gamma_9+\gamma_4+\gamma_7<\gamma_0,\ 3\gamma_9+\gamma_7<\gamma_1,\ 6\gamma_9+\gamma_5+\gamma_7<\gamma_0,\ \gamma_8<\gamma_9$, we have
		\begin{align*}
		\xi_{1n}&=O(n^{\gamma_4+3\gamma_8-\gamma_0})+O(n^{3\gamma_8-\gamma_1})+O(n^{\gamma_5+6\gamma_8-\gamma_0}),\\
		&<O(n^{\gamma_4+3\gamma_9-\gamma_0})+O(n^{3\gamma_9-\gamma_1})+O(n^{\gamma_5+6\gamma_9-\gamma_0}) \to 0,\\
		\xi_{2n}&=O(n^{\gamma_4+3\gamma_9-\gamma_0})+O(n^{3\gamma_9-\gamma_1})+O(n^{\gamma_5+6\gamma_9-\gamma_0}) \to 0,
		\end{align*}
		and
		\begin{align*}
		C_{3n}\xi_{1n}&=O(n^{\gamma_6+\gamma_4+3\gamma_8-\gamma_0})+O(n^{\gamma_6+3\gamma_8-\gamma_1})+O(n^{\gamma_6+\gamma_5+6\gamma_8-\gamma_0}) \to 0,\\
		C_{4n}\xi_{2n}&=O(n^{\gamma_7+\gamma_4+3\gamma_9-\gamma_0})+O(n^{\gamma_7+3\gamma_9-\gamma_1})+O(n^{\gamma_7+\gamma_5+6\gamma_9-\gamma_0}) \to 0,
		\end{align*}
		so
		$\rho_{1n}+\rho_{2n} \to 0$. Combining these results and $\eta_n=O(n^{-\gamma_1}),\ \gamma_1>0$, we have
		\begin{align*}
		&12\tau^2-12\left(M-p_0+1\right)\eta_n^2-4\eta_n^2-4\eta_n\left(C_{4n}+C_{3n}\right)-\Delta_{p_0,n}V_{1n}-\Delta_{p_0+M-1,n}V_{2n}-\rho_{1n}-\rho_{2n} \\
		&\to 12\tau^2,
		\end{align*}
		which implies Assumption 4.
	\end{proof}
	
	\begin{proof}{(Proof of Theorem 4)} By Assumption 6, we can take $\gamma_0$ satisfying $0<\gamma_0<\tilde{\gamma}_0,\ \gamma_2<\gamma_0,\ \gamma_3<{\gamma}_0,\ \gamma_4<{\gamma}_0,\
		\gamma_5<{\gamma}_0,\ \gamma_7<\gamma_1,\  \gamma_6<\gamma_1,\ \gamma_4+\gamma_6<{\gamma}_0,\ \gamma_5+\gamma_6<{\gamma}_0,\ \gamma_4+\gamma_7<{\gamma}_0$, and $\gamma_5+\gamma_7<{\gamma}_0$.
		
		By Lemma \ref{lemma del}, we have $3(M+1)^2(9M+1)\Delta_{2,n}=O(n^{-\gamma_0}) \to 0$ in probability. Let $\alpha=1/2$ in Assumption 3. and Assumption 3 holds with a probability tending to one.		
		Furthermore, by the conditions on $\gamma_0,\ldots,\gamma_7$ and $\gamma_{10}$ and Lemma \ref{lemma del}, some algebra yields
		\begin{align*}
		\rho_{1n},\ \rho_{2n},\ 4\eta_n^2, \ 12(M-p_0+1)\eta_n^2,\ 4\eta_n(C_{4n}+C_{3n}),\ \Delta_{p_0,n}V_{1n}+\Delta_{p_0+M-d,n}V_{2n}  \to 0,
		\end{align*}in probability, which implies $12\tau^2-12\left(M-p_0+1\right)\eta_n^2-4\eta_n^2-4\eta_n\left(C_{4n}+C_{3n}\right)-\Delta_{p_0,n}V_{1n}-\Delta_{p_0+M-1,n}V_{2n}-\rho_{1n}-\rho_{2n} \to 12\tau^2$ in probability, and then Assumption 4 holds with a probability tending to one. By Theorem 2, the proof is completed.		
	\end{proof}

\subsection{Proof when using the nonparametric additive model}
	\begin{proof} In Section 4, we state that there exists $f$ such that the screening method based on the nonparametric additive model (14) does not work. Here we prove this by giving an example of such $f$.

Consider the simple case of $p=2$ and the function $f(x_1,x_2)=(x_1-1/2)(x_2-1/2)$. Note that $\int_{[0,1]^2}f(x_1,x_2)d\v{x}=0$. Then the best additive approximation of $f$ is
\begin{align*}
b_1(x_1)+b_2(x_2) =\mathop{\arg\min}_{s_1\in C[0,1], \ s_2 \in C[0,1]} \int_{[0,1]^2}\left[f(x_1,x_2)-\left( s_1(x_1)+s_2(x_2) \right)^2  \right] d\v{x}.
\end{align*}By noting that $\int_{[0,1]^2}s_1(x_1)f(x_1,x_2)d\v{x}=\int_{[0,1]^2}s_1(x_1)(x_1-1/2)(x_2-1/2)d\v{x}=0$ and $\int_{[0,1]^2}s_2(x_2)f(x_1,x_2)d\v{x}=0$,
we have\begin{align*}
& \int_{[0,1]^2}\left[f(x_1,x_2)-\left( s_1(x_1)+s_2(x_2) \right)^2  \right] d\v{x}\\&=\int_{[0,1]^2}f(x_1,x_2)^2d\v{x}
+\int_{[0,1]^2}\left( s_1(x_1)+s_2(x_2)\right)^2d\v{x} -2\int_{[0,1]^2}\left(s_1(x_1)+s_2(x_2) \right) f(x_1,x_2)d\v{x}
\\&=\int_{[0,1]^2}f(x_1,x_2)^2d\v{x}
+\int_{[0,1]^2}\left( s_1(x_1)+s_2(x_2)\right)^2d\v{x}\geqslant\int_{[0,1]^2}f(x_1,x_2)^2d\v{x},
\end{align*} which implies $b_1=b_2=0$. This indicates that the nonparametric additive model cannot identify the two active variables $x_1$ and $x_2$ of $f$.
	\end{proof}

	\section*{Acknowledgement}
	
The authors thank the associate editor for helpful and constructive comments. This work is supported by the National Natural Science Foundation of China (Grant No. 11671386, 11871033) and Key Laboratory of Systems and Control, CAS.

	\section*{References}
	\bibliographystyle{unsrt}
	\begin{description}
		\item {}
		Bertsimas, D., King, A., and Mazumder, R. (2016), ``Best Subset Selection via a Modern Optimization Lens," \textit{The Annals of Statistics}, 44, 813-852.
		\item {}
		Chen, X. S., and Li, W. (2003),  ``Relative Perturbation Bounds of Eigenvalues for Positive Definite Hermite Matrices,"  \textit{Chinese Journal of Engineering Mathematics}, 69, 140-142.
		\item{}
		Clarke, J. A. (2001), ``Energy Simulation in Building Design (2nd ed.)," Butterworth-Heinemann, Oxford.
		\item {}
		Fan, J., and Lv, J. (2008), ``Sure Independence Screening for Ultrahigh Dimensional Feature Space, \textit{Journal of the Royal Statistical Society, Series B}," 70, 849-911.
		\item {}
		Fan, J., and Song, R. (2010), ``Sure Independence Screening in Generalized Linear Models With NP-Dimensionality," \textit{The  Annals of Statistics}, 38, 3567-3604.
		\item{}
		Fan, J., Feng, Y., and Song, R. (2011), ``Nonparametric Independence Screening in Sparse Ultra-High Dimensional Additive Models," \textit{Journal of the American Statistical Association}, 106, 544-557.
		\item{}
		Fan, J., Ma, Y., and Dai, W. (2014), `` Nonparametric Independence Screening in Sparse Ultra-High Dimensional Varying Cofficient Models," \textit{Journal of the American Statistical Association}, 109, 1270-1284.
		\item{}
		Fang, K. T., Li, R., and Sudjianto, A. (2006), ``Design and Modelling for Computer Experiments," Chapman \& Hall/CRC.
        \item{}
		Golub, G., and Van Loan, C. F. (1996), ``Matrix Computations,"  3rd Ed., Baltimore, The Johns Hopkins University Press.
		\item{}
		Golub, G. H., Heath, M., and Wahba, G. (1979), ``Generalized Cross-Validation as a Method for Choosing a Good Ridge Parameter," \textit{Technometrics}, 21, 215-223.
        \item{}
		Heinrich, S., Novak, E., Wasilkowski, G. W., and Wozniakowski, H. (2001), ``The Inverse of the Star-discrepancy Depends Linearly on the Dimension," \textit{Acta Arithmetica}, 96, 279-302.
		\item{}
		Huang, Q., and Zhu, Y. (2016), ``Model-Free Sure Screening via Maximum Correlation," \textit{Journal of Multivariate Analysis}, 148, 89-106.
		\item{}
		Jahangirian, M., Eldabi, T., Naseer, A., Stergioulas, L. K., and Young, T. (2010), ``Simulation in Manufacturing and Business: A review," \textit{European Journal of Operational Research}, 203, 1-13.
		\item{}
		Kiefer, J. and Wolfowitz, J. (1958), ``On the Deviations of the Empiric Distribution Function of Vector Chance Variables," \textit{ Transactions of the American Mathematical Society}, 87, 173-186.		
        \item{}
		Li, R., Zhong, W., and Zhu L. (2012), ``Feature Screening via Distance Correlation Learning," \textit{Journal of American Statistical Association}, 107, 1129-1139.
		\item{}
		Linkletter, C., Bingham, D., Hengartner, N., Higdon, D. and Ye, K. Q. (2006). ``Variable Selection for Gaussian Process Models in Computer Experiments," \textit{Technometrics}, 48, 478-490.
		\item {}
		Lu, J., and Lin, L. (2017), ``Model-free conditional screening via conditional distance correlation," \textit{Statistical Papers}, https://doi.org/10.1007/s00362-017-0931-7.
		\item{}
		Matheron, G. (1963), ``Principles of Geostatistics," \textit{Economic Geology}, 58, 1246-1266.
		\item{}
		Moon, H., Dean, A. M. and Santner, T. J. (2012), ``Two-stage Sensitivity-based Group Screening in Computer Experiments," \textit{Technometrics}, 54, 376-387.
		\item{}
		Morris, M. D. (1991), ``Factorial Sampling Plans for Preliminary Computational Experiments," \textit{Technometrics}, 33, 161-174.
		\item{}
		Morris, M. D., Mitchell, T. J., and Ylvisaker, D. (1993), ``Bayesian Design and Analysis of Computer Experiments: Use of Derivatives in Surface Prediction," \textit{Technometrics}, 35, 243-255.
		\item{}
		Owen, A. B. (2005), ``Multidimensional Variation for Quasi-Monte Carlo," in ``Fan, J. and Li, G., editors, International Conference
		on Statistics in honour of Professor Kai-Tai Fang's 65th birthday."
		\item{}
		Reich, B. J., Storlie, C. B., and Bondell, H. D. (2009), ``Variable Selection in Bayesian Smoothing Spline Anova Models: Application to Deterministic Computer Codes," \textit{Technometrics}, 51, 110-120.
		\item{}
		Roulstone, L. and Norbury, J. (2013), ``Invisible in the Storm: the Role of Mathematics in Understanding Weather". Princeton University Press.
\item{}
Ruppert, D., Wand, M. P., and Carroll, R. J. (2003). ``Semiparametric Regression". Cambridge University Press.
		\item{}	
		Santner, T. J., Williams, B. J., and Notz, W. I.  (2018), ``The Design and Analysis of Computer Experiments," The Second Edition, Springer-Verlag, New York.
		\item{}
		Schonlau, M., and Welch, W. J. (2006), ``Screening the Input Variables to a Computer Model via Analysis of Variance and Visualization," in Screening Methods for Experimentation in Industry, Drug Discovery and Genetics, eds. A. M. Dean and S. Lewis, New York: Wiley, 308–327.
		\item{}
		Shen, X. Pan, W., Zhu, Y., and Zhou, H. (2013), ``On Constrained and Regularized High-dimensional regression," \textit{Annals of the Institute of Statistical Mathematics}, 65, 807--832.
		\item{}
		Sobol', I. M. and Saltelli, A. (1995), ``About the Use of Rank Transformation in Sensitivity Analysis of Model Output," \textit{Reliability Engineering System Safety}, 50, 225-239.
\item{}
Sung, C.-L., Wang, W. J., Plumlee, M., and Haaland, B. (2017), ``Multi-Resolution Functional ANOVA for Large-Scale, Many-Input Computer Experiments," \textit{arXiv:1709.07064}.

		\item{}
		Tezuka, S. (2002), ``Quasi-Monte Carlo Discrepancy between Theory and Practice," \textit{ Monte Carlo and Quasi-Monte Carlo Methods 2000}, 124-140.
		\item{}
		Tibshirani, R. (1996), ``Regression Shrinkage and Selection via  Lasso," \textit{Journal of the Royal Statistical Society, Ser. B}, 58, 267-288.
		\item{}
		Worley, B. A. (1987),  ``Deterministic Uncertainty Analysis," Technical Report ORNL-6428, Oak Ridge National Research Laboratory.
		\item{}
		Xiong, S. (2014), ``Better Subset Regression," \textit{Biometrika}, 101, 71-84.
		\item{}
		Xiong, S., Qian, P. Z. G., and Wu, C. F. J. (2013). ``Sequential Design and Analysis of High-Accuracy and Low-Accuracy Computer Codes," \textit{Technometrics}, 55, 37-46.
		\item{}
		Xu, C. and Chen J. (2014), ``The Sparse MLE for Ultra-High-Dimensional Feature Screening," \textit{Journal of American Statistical Association}, 109, 1257--1269.
		
		\item{}
		Yang, X.-S. (2010), ``Appendix A: Test Problems in Optimization," \textit{Engineering optimization}, ed. X.-S. Yang, John Wiley and Sons, Inc..
		\item {}
		Zhu, L. P., Li, L. X., Li, R. Z., and Zhu L. X. (2011), `` Model-Free Feature Screening for Ultrahigh Dimensional data,"  \textit{Journal of American Statistical Association}, 106, 1464-1475.
	\end{description}
\end{document}